\documentclass[11pt]{article}

\usepackage{clrscode}
\usepackage[T1]{fontenc}
\usepackage{amsmath}
\usepackage{amsthm}
\usepackage{amssymb}
\usepackage{epsfig}
\usepackage{enumerate}
\usepackage{bm}

\newcommand{\spl}{\mathrel{L}}

\newcommand{\shortcite}[1]{\cite{#1}}

\newcommand{\naturals}{\mathbb{N}}
\newcommand{\integers}{\mathbb{Z}}

\newcommand{\fixenumerate}{\addtolength{\itemsep}{-4pt}}

\newcommand{\calC}{{{\mathcal{C}}}}
\newcommand{\calD}{{{\mathcal{D}}}}

\newcommand{\Pclass}{\textnormal{P}}
\newcommand{\FPclass}{\textnormal{FP}}
\newcommand{\NPclass}{\textnormal{NP}}

\newcommand{\sharpPclass}{\textnormal{\#P}}

\newcommand{\np}{\NPclass}
\newcommand{\sharpp}{\sharpPclass}
\newcommand{\fp}{\FPclass}
\newcommand{\p}{\Pclass}

\newcommand{\calS}{{{\mathcal{S}}}}

\newcommand{\score}[1]{\id{score}_{#1}}
\newcommand{\N}[1]{N_{#1}}

\newcommand{\revnot}[1]{\overleftarrow{#1}}

\theoremstyle{plain}
\newtheorem{theorem}{Theorem}[section]

\newtheorem{corollary}[theorem]{Corollary}
\newtheorem{definition}[theorem]{Definition}
\newtheorem{lemma}[theorem]{Lemma}

\newcommand{\scores}{{\mathit{scores}}}
\newcommand{\approval}{{\mathrm{approval}}}
\newcommand{\cutoff}[2]{{\mathrm{cutoff}(#1,#2)}}
\newcommand{\best}{{\mathrm{top}}}
\newcommand{\rank}{{\mathrm{rank}}}

\title{Possible Winners in Noisy Elections\footnote{An early version
    of this paper was presented at the Twenty-Sixth AAAI Conference on
    Artificial Intelligence (AAAI-12), July 22--26, 2012, in Toronto,
    Canada and at IJCAI Workshop on Social Choice and Artificial
    Intelligence, July 16, 2011, in Barcelona.}}  \author{
  Krzysztof Wojtas\quad Krzysztof Magiera\quad Tomasz Mi\k{a}sko\quad Piotr Faliszewski \\
  AGH University of Science and Technology\\
  Krakow, Poland }

\begin{document}

\maketitle

\begin{abstract}
  We consider the problem of predicting winners in elections, for the
  case where we are given complete knowledge about all possible
  candidates, all possible voters (together with their preferences),
  but where it is uncertain either which candidates exactly register
  for the election or which voters cast their votes. Under reasonable
  assumptions, our problems reduce to counting variants of election
  control problems.  We either give polynomial-time algorithms or
  prove $\sharpp$-completeness results for counting variants of
  control by adding/deleting candidates/voters for Plurality,
  $k$-Approval, Approval, Condorcet, and Maximin voting rules.  We
  consider both the general case, where voters' preferences are
  unrestricted, and the case where voters' preferences are
  single-peaked.
\end{abstract}

\section{Introduction}

Predicting election winners is always an exciting activity: Who will
be the new president? Will the company merge with another one? Will
taxes be higher or lower? The goal of this paper is to establish the
computational complexity of a family of problems modeling a certain
type of winner-prediction problems.

Naturally, predicting winners is a hard task, full of
uncertainties. For example, we typically are not sure which voters
will eventually cast their votes or, sometimes, even which candidates
will participate in the election (consider, e.g., a candidate
withdrawing due to personal reasons).  Further, typically we do not
have complete knowledge regarding each voters' preferences.
Nonetheless, elections are in everyday use both among humans
(consider, e.g., political elections, elections among companies'
shareholders, various polls on the Internet and social
media, or even sport events, where judges ``vote'' on who is the best
competitor) and among software agents (see, e.g., election
applications for planning in multiagent
systems~\cite{eph-ros:j:multiagent-planning}, for recommendation
systems~\cite{gho-mun-her-sen:c:voting-for-movies,bou-lu:c:chamberlin-courant,sko-fal-sli:c:multiwinner},
for web-search~\cite{dwo-kum-nao-siv:c:rank-aggregation}, or natural
language processing~\cite{kut-kit:j:nlp}) and,
thus, the problem of predicting election winners is far too important
to be abandoned simply because it is difficult.

In this paper, we focus on a variant of the winner-prediction problem
where we have complete knowledge regarding all possible candidates and
all eligible voters (including knowledge of their
preferences\footnote{Note that while full-knowledge assumption
  regarding voters' preferences might seem very unrealistic, it is
  standard within computational social choice literature, and in our
  case can often be justified (e.g., election polls can provide a good
  approximation of such knowledge).}), but we are uncertain as to
which candidates and which voters turn up for the actual election (see
Section~\ref{sec:related} for other approaches to the
problem). However, modelling uncertainty regarding both the candidate
set and the voter collection on one hand almost immediately leads to
computationally hard problems for typical election systems, and on the
other hand does not seem to be as well motivated as focusing on each
of these sets separately. Thus, we consider the following two
settings:
\begin{enumerate}
\fixenumerate
\item The set of candidates is fixed, but for each possible subset
  of voters we are given a probability that exactly these voters
  show up for the vote.
\item The set of voters is fixed, but for each possible subset of
  candidates we are given a probability that exactly these
  candidates register for the election.
\end{enumerate}
The former setting, in particular, corresponds to political elections
(e.g., to presidential elections), where the candidate set is
typically fixed well in advance due to election rules, the set of all
possible voters (i.e., the set of all citizens eligible to vote) is
known, but it is not clear as to which citizens choose to cast their
votes.  The latter setting may occur, for example, if one considers
software agents voting on a joint plan~\cite{eph-ros:j:multiagent-planning}. The set of agents participating in the
election is typically fixed, but various variants of the plan can be
put forward or dismissed dynamically.  In either case, our goal is to
compute each candidate's probability of victory.
However, our task would very quickly become computationally
prohibitive (or, difficult to represent on a computer) if we allowed
arbitrary probability distributions. Thus, we have to choose some
restriction on the distributions we consider. Let us consider the
following example.  

Let the set $C$ of candidates participating in the election be fixed
(e.g., because the election rules force all candidates to register
well in advance). We know that some set $V$ of voters will certainly
vote (e.g., because they have already voted and this information is
public\footnote{Naturally, in typical political elections such
  information would not be public and we would have to rely on
  polls. On the other hand, in multiagent systems there can be cases
  where votes are public.}). The set of voters who have not decided to
vote yet is $W$. From some source (e.g., from prior experience) we
have a probability distribution $P$ on the number of voters from $W$
that will participate in the election (we assume that each equal-sized
subset from $W$ is equally likely to joint the election; we have no
prior knowledge as to which eligible voters are more likely to vote).
That is, for each $i$, $0 \leq i \leq \|W\|$, by $P(i)$ we denote the
probability that exactly $i$ voters from $W$ join the election (we
assume that $P(i)$ is easily computable). By $Q(i)$ we denote the
probability that a certain designated candidate $p$ wins provided that
exactly $i$, randomly chosen, voters from $W$ participate in the
election.  The probability that $p$ wins is given by:
\[
P(0)Q(0) + P(1)Q(1) + \cdots P(\|W\|)Q(\|W\|).
\]
We use this formula to compute the probability of each candidate's
victory, which gives some idea as to who is the likely winner of the
election. 

To compute $Q(i)$, we have to compute for how many subsets $W'$ of $W$
of size exactly $i$ candidate $p$ wins after adding the voters from
$W'$ to the election, and divide it by $\|W\| \choose i$.  That is,
computing $Q(i)$ boils down to solving a counting variant of control
by adding voters problem.  In the decision variant of this control
problem, introduced by Bartholdi, Tovey, and
Trick~\cite{bar-tov-tri:j:control}, we are given an election where
some voters have already registered to vote, some are
not-yet-registered, and we ask if it is possible to add some number of
these not-yet-registered voters (but no more than a given limit) to
ensure that a designated candidate wins. In the counting variant,
studied first in this paper, we ask how many ways are there to add
such a group of voters.

One can do reasoning analogous to the one presented above for the case
of control by deleting voters and for the case of control by
adding/deleting candidates. That is, our winner prediction problems,
in essence, reduce to solving counting variants of election control
problems. Our goal in this paper is, thus, to study the computational
complexity of these counting control problems.

Computational study of election control was initiated by Bartholdi,
Tovey, and Trick~\shortcite{bar-tov-tri:j:control}, and in recent
years was continued by a number of researchers (we discuss related
work in Section~\ref{sec:related}; we also point the reader to the
survey of Faliszewski, Hemaspaandra, and
Hemaspaandra~\cite{fal-hem-hem:j:cacm-survey} for a more detailed
discussion of the complexity of election control). However, our paper
is the first one to study counting variants of control (though, as we
discuss in Section~\ref{sec:related}, the papers of Walsh and
Xia~\cite{wal-xia:c:lot-based} and of Bachrach, Betzler, and
Faliszewski~\cite{bac-bet-fal:c:counting-pos-win} are very close in
spirit).  

It is well-known that election control problems tend to have fairly
high worst-case complexity. Indeed, we are not aware of a single
practical voting rule for which the decision variants of the four most
typical election control problems (the problems of adding/deleting
candidates/voters) are all polynomial-time solvable. Naturally, when a
decision variant of a counting problem is $\np$-hard, then we cannot
hope that its counting variant would be polynomial-time solvable. In
effect, from the technical perspective, our research can be seen as
answering the following question: For which voting rules and for which
control types is the counting variant of the control problem no harder
than its decision variant? From the practical perspective, it is,
nonetheless, more important to simply have effective algorithms. To
this end, we follow Faliszewski et
al.~\cite{fal-hem-hem-rot:j:single-peaked-preferences} and in addition
to the general setting where each voter can cast any possible vote, we
study the single-peaked case, where the voting is restricted to a
certain class of votes viewed as ``reasonable'' (intuitively, in the
single-peaked case we assume that the candidates are located on a
one-dimensional spectrum of possible opinions, each voter has some
favorite opinion on this spectrum, and the voters form their
preferences based on candidates' distances from their favorite
position). Single-peaked preferences, introduced by
Black~\cite{bla:b:polsci:committees-elections} over 50 years ago, are
often viewed as a natural (if somewhat simplified due to its
one-dimensional nature) model of how realistic votes might look like
(however, whenever one makes a statement of this form, one should keep
in mind the results of Sui, Francois-Nienaber, and
Boutilier~\cite{sui-fra-bou:c:single-peaked}). It turns out that for
the sinlge-peaked case, we obtain polynomial-time algorithms for
counting variants of all our control problems and all the rules that
we consider (except for the Approval rule, for which we have no
results in this case).

We mention that our model of winner prediction, based on counting
variants of control, is similar to, though more general than, the
model where we assume that each voter casts his or her vote with some
probability $p$ (universal for all the voters).  We can simulate this
scenario in our model by providing an appropriate function $P$.
Specifically, if there are $n$ voters and the probability that each
particular voter $v$ casts his or her vote is $p$, then the
probability that exactly $i$ voters vote is given by the binomial
distribution:
\[
   P(i) = {n \choose i}p^i(1-p)^{n-i}.
\]
On the other hand, our model does not capture the situation where each
voter $v$ has a possibly different probability $p_v$ of casting a
vote. We believe that this latter model deserves study as well, but we
do not focus on it in this paper.

The paper is organized as follows. First, in Section~\ref{sec:prelim},
we formally define elections and provide brief background on
complexity theory (focusing on counting problems). Then, in
Section~\ref{sec:problems}, we define counting variants of election
control problems and prove some of their general properties.  Our main
results are in Section~\ref{sec:results}. There we study the
complexity of Plurality, $k$-Approval, Appoval, Condorcet rule, and
Maximin rule. For each of the rules, we consider the unrestricted case
and, if we obtain a hardness result, we consider the single-peaked
case (except for Approval, for which we were not able to obtain
results for the single-peaked case).  Finally, we discuss related work
in Section~\ref{sec:related} and bring together our conclusions and
discuss future work in Section~\ref{sec:conclusions}.

\section{Preliminaries}
\label{sec:prelim}

\noindent\textbf{Elections.}\quad
An \emph{election} $E$ is a pair $(C,V)$ such that $C$ is a finite set
of candidates and $V$ is a finite collection of voters.  We typically
use $m$ to denote the number of candidates and $n$ to denote the
number of voters.  Each voter has a preference order in which he or
she ranks candidates, from the most desirable one to the most despised
one.  For example, if $C = \{a,b,c\}$ and a voter likes $b$ most and
$a$ least, then this voter would have preference order $b > c >
a$. (However, under Approval voting, instead of ranking the candidates
each voter simply indicates which candidates he or she approves of.)
We sometimes use the following notation.  Let $A$ be some subset of
the candidate set. Putting $A$ in a preference order means listing
members of $A$ in lexicographic order and putting $\overleftarrow{A}$
in a preference order means listing members of $A$ in the reverse of
the lexicographic order. For example, if $C = \{a,b,c,d\}$ then $a > C
- \{a,b\} > d$ means $a > c > d > b$ and $a > \overleftarrow{C -
  \{a,b\}} > d$ means $a > d > c > b$.
Given an election $E = (C,V)$, we write $\N{E}(c,c')$ to denote the
number of voters in $V$ that prefer $c$ to $c'$. The function $N_E$ is
sometimes referred to as the \emph{weighted majority graph} of
election $E$.

In general, given a candidate set $C$, the voters are free to report
any preference order over $C$ (we refer to this setting as the
\emph{unrestricted case}). However, it is often more realistic to
assume some restriction on the domain of possible votes (for example,
in real-life political elections we do not expect to see many voters
who rank the extreme left-wing candidate first, then the extreme
right-wing candidate, and so on, in an interleaving fashion).  Thus,
in addition to the unrestricted case, we also study the case of
\emph{single-peaeked} preferences of
Black~\cite{bla:b:polsci:committees-elections} (this domain
restriction is quite popular in the computational social choice
literature, and was already studied, e.g., by
Conitzer~\cite{con:j:eliciting-singlepeaked},
Walsh~\cite{wal:c:uncertainty-in-preference-elicitation-aggregation},
Faliszewski et al.~\cite{fal-hem-hem-rot:j:single-peaked-preferences},
and others; see Section~\ref{sec:related}).

\begin{definition}\label{def:sp}
  Let $C$ be a set of candidates and let $L$ be a linear order over
  $C$ (we refer to $L$ as the societal axis). We say that a preference
  order $>$ (over $C$) is single-peaked with respect to $L$ if for
  each three candidate $a, b, c \in C$, it holds that
  \[ 
    (a \spl b \spl c \lor  c \spl b \spl a) \implies (a > b \implies b > c)
  \]
  An election $E = (C,V)$ is single-peaked with respect to $L$ if each
  vote in $V$ is single-peaked with respect to $L$. An election $E =
  (C,V)$ is single-peaked if there is a societal axis $L$ such that
  $E$ is single-peaked with respect to $L$.
\end{definition}

There are polynomial-time algorithms that given an election $E$ verify
if it is single-peaked and, if so, provide the appropriate societal
axis witnessing this
fact~\cite{bar-tri:j:stable-matching-from-psychological-model,esc-lan-ozt:c:single-peaked-consistency,bal-har:j:characterization-single-peaked}. Thus,
following
Walsh~\cite{wal:c:uncertainty-in-preference-elicitation-aggregation},
whenever we consider problems about single-peaked elections, we assume
that we are given appropriate societal axis as part of the input (if
it were not provided, we could always compute it.)\bigskip

\noindent\textbf{Voting Systems.}\quad
A \emph{voting system} (a \emph{voting rule}) is a rule which
specifies how election winners are determined. We allow an election to
have more than one winner, or even to not have winners at all. In
practice, tie-breaking rules are used, but here we disregard this
issue by simply using the unique winner model (see
Definition~\ref{def:control}). However, we point the reader
to~\cite{obr-elk-haz:c:ties-matter,obr-elk:c:random-ties-matter,obr-zic-elk:c:multiwinner-tie-breaking}
for a discussion regarding the influence of tie-breaking for the case
of election manipulation problems
(see~\cite{fal-hem-hem:j:cacm-survey,fal-pro:j:manipulation,bra-con-end:b:comsoc} for
overviews of the manipulation problem specifically and computational aspects
of voting generally).

We consider the following voting rules (in the description below we 
assume that $E = (C,V)$ is an election with $m$ candidates and $n$
voters):
\begin{description}

\item[Plurality.] Under Plurality, each candidate receives a point for
  each vote where this candidate is ranked first. The candidates with
  most points win.

\item[$\bm{k}$-Approval.] For each candidate $c \in C$, we define
  $c$'s $k$-Approval score, $\score{E}^k(c)$, to be the number of
  voters in $V$ that rank $c$ among the top $k$ candidates; the
  candidates with highest scores win. Note that Plurality is, in
  effect, a nickname for $1$-Approval. (We mention that $k$-Veto means
  $(m-k)$-Approval, though we do not study $k$-Veto in this paper).

\item[Approval.] Under Approval (without the qualifying ``$k$-''), the
  score of a candidate $c \in C$, $\score{E}^a(c)$, is the number of
  voters that approve of $c$ (recall that under Approval the voters do
  not cast preference orders but 0/1 approval vectors, where for each
  candidate they indicate if they approve of this candidate or
  not). Again, the candidates with highest scores win. (Note that the
  notion of single-peaked elections that we have provided as
  Definition~\ref{def:sp} does not apply to preferences specified as
  approval vectors; Faliszewski et
  al.~\cite{fal-hem-hem-rot:j:single-peaked-preferences} have provided
  a natural variant of single-peakedness for approval vectors, but
  since we did not obtain results for this case, we omitted this
  definition.)

\item[Condorcet.] A candidate $c$ is a \emph{Condorcet winner} exactly
  if $\N{E}(c,c') > \N{E}(c',c)$ for each $c' \in C - \{c\}$.  A
  candidate $c$ is a \emph{weak Condorcet winner} exactly if
  $\N{E}(c,c') \geq \N{E}(c',c)$ for each $c' \in C - \{c\}$. Note
  that if an election has a Condorcet winner, then this winner is
  unique (though, an election may have several weak Condorcet
  winners).\footnote{We mention that sometimes Condorcet's rule is not
    considered a voting rule (but, for example, a consensus
    notion~\cite{mes-nur:b:distance-realizability}) because there are
    elections for which there are no Condorcet winners. However,
    Condorcet's rule has been traditionally studied in the context of
    election control, and--thus---we study it to (a) enable comparison
    with other papers, and (b) because for single-peaked elections the
    results for Condorcet rule translate to all Condorcet-consistent
    rules.}  We recall the classic result that if the voters are
  single-peaked then the election always has (weak) Condorcet
  winner(s) (if the number of voters is odd, then the election has a
  unique Condorcet winner).

\item[Maximin.] Maximin rule is defined as follows. The score of a
  candidate $c$ is defined as $\score{E}^m(c) = \min_{d \in C
    -\{c\}}N_E(c,d)$. The candidates with highest Maximin score are
  Maximin winners.
\end{description}

Maximin is an example of a so-called Condorcet-consistent rule.  A
rule $R$ is Condorcet-consistent if the following holds: If $E$ is an
election with (weak) Condorcet winner(s), then the $R$-winners of $E$
are exactly these (weak) Condorcet winner(s).  We also considered
studying the Copeland rule (see, e.g., the work of Faliszewski et
al.~\cite{fal-hem-hem-rot:j:llull} regarding the complexity of control
for the Copeland rule), but for this rule all relevant types of
control are $\np$-complete and, so, at best we would obtain
$\sharpp$-completeness results based on simple generalizations of
proofs from the literature. Nonetheless, since Copeland's rule is
Condorcet-consistent, results for the single-peaked case are valid for
it.

\bigskip

\noindent\textbf{Notation for Graphs.}\quad We assume familiarity with basic
concepts of graph theory. Given an undirected graph $G$, by $V(G)$ we
mean its set of vertices, and by $E(G)$ we mean its set of
edges.\footnote{We use this slightly nonstandard notation to be able
  to also use symbols $V$ and $E$ to denote voter collections and
  elections.}  Whenever we discuss a bipartite graph $G$, we assume
that $V(G)$ is partitioned into two subsets, $X$ and $Y$, such that
each edge connects some vertex from $X$ with some vertex from $Y$. We
write $X(G)$ to denote $X$, and $Y(G)$ to denote $Y$.\bigskip

\noindent\textbf{Computational Complexity.}\quad
We assume that the reader is familiar with standard notions of
computational complexity theory, as presented, for example, in the
textbook of Papadimitriou~\cite{pap:b:complexity}, but we
briefly review notions regarding the complexity theory of counting
problems. Let $A$ be some computational problem where for each
instance $I$ we ask if there exists some mathematical object
satisfying a given condition. In the counting variant of $A$, denoted
$\#A$, for each instance $I$ we ask for the number---denoted
$\#A(I)$---of the objects that satisfy the condition.
For example, consider the following problem.

\begin{definition}
  An instance of X3C is a pair $(B,\calS)$, where $B = \{b_1, \ldots,
  b_{3k}\}$ and $\calS = \{S_1, \ldots, S_n\}$ is a family of
  $3$-element subsets of $B$. In X3C we ask if it is possible to find
  exactly $k$ sets in $\calS$ whose union is exactly $B$. In \#X3C we
  ask how many $k$-element subsets of $\calS$ have $B$ as their union.
\end{definition}

The class of counting variants of $\np$-problems is called $\sharpp$
and the class of functions computable in polynomial time is called
$\fp$.  To reduce counting problems to each other, we will use one of
the following reducibility notions.

\begin{definition}
  Let $\#A$ and $\#B$ be two counting problems. 
  \begin{enumerate} 
\fixenumerate
\item We say that $\#A$ Turing reduces to $\#B$ if there exists an
  algorithm that solves $\#A$ in polynomial time given oracle access
  to $\#B$ (i.e., given the ability to solve $\#B$ instances in unit
  time).
  \item We say that $\#A$ metric reduces to $\#B$ if there exist two
    polynomial-time computable functions, $f$ and $g$, such that for
    each instance $I$ of $\#A$ it holds that (1) $f(I)$ is an instance
    of $\#B$, and (2) $\#A(I) = g(I, \#B(f(I)) )$.
  \item We say that $\#A$ parsimoniously reduces to $\#B$ if $\#A$
    metric reduces to $\#B$ via functions $f$ and $g$ such that for
    each instance $I$ and each integer $k$, $g(I,k) = k$ (i.e., there
    is a polynomial-time computable function $f$ such that for each
    instance $I$ of $\#A$, we have $\#A(I) = \#B(f(I))$).
  \end{enumerate}
\end{definition}

It is well-known that these reducibility notions are transitive.  For
a given reducibility notion $R$, we say that a problem is
$\sharpp$-$R$-complete if it belongs to $\sharpp$ and every
$\sharpp$-problem $R$ reduces to it.  For example, \#X3C is
$\sharpp$-parsimonious-complete~\cite{hun-mar-rad-ste-:j:planar-counting-problems}.
Throughout this paper we will write $\sharpp$-complete to mean
$\sharpp$-parsimonious-complete. Turing reductions were used, e.g., by
Valiant~\shortcite{val:j:permanent} to show $\sharpp$-hardness of
computing a permanent of a 0/1 matrix. As a result, he also showed
$\sharpp$-Turing-hardness of the following problem. 
\begin{definition}
  In \#PerfectMatching we are given a bipartite graph $G =
  (G(X),G(Y),G(E))$ with $\|G(X)\| = \|G(Y)\|$ and we ask how many
  perfect matchings does $G$ have.
\end{definition}
Zank{\'{o}}~\cite{zan:j:sharp-p} improved upon this result by showing
that \#PerfectMatching is \#P-many-one-complete (``many-one'' is yet
another reducibility type, more general than parsimonious reductions
but less general than metric reductions). From our perspective, it
suffices that, in effect, \#PerfectMatching is \#P-metric-complete.
Metric reductions were introduced by
Krentel~\shortcite{kre:j:optimization}, and parsimonious reductions
were defined by Simon~\shortcite{sim:thesis:complexity}.

\section{Counting Variants of Control Problems}
\label{sec:problems}
In this section we formally define counting variants of election
control problems and show some interconnections between some of
them. We are interested in control by adding candidates (AC), control
by deleting candidates (DC), control by adding voters (AV), and
control by deleting voters (DV). For each of these problems, we
consider its constructive variant (CC) and its destructive variant
(DC).

\begin{definition}\label{def:control}
  Let $R$ be a voting system. In each of the counting variants of
  constructive control problems we are given a candidate set $C$, a
  voter collection $V$, a nonnegative integer $k$, and a designated
  candidate $p \in C$. In constructive control by adding voters we are
  additionally given a collection $W$ of unregistered voters, and in
  constructive control by adding candidates we are additionally given
  a set $A$ of unregistered candidates. In these problems we ask for
  the following quantities:
  \begin{enumerate}
\fixenumerate
  \item In control by adding voters ($R$-\#CCAV),
    we ask how many sets $W'$, $W' \subseteq W$, are there such that
    $p$ is the unique winner of $R$-election $(C,V \cup W')$, where
    $\|W'\| \leq k$.
  \item In control by deleting voters ($R$-\#CCDV), we ask how many
    sets $V'$, $V' \subseteq V$ are there such that $p$ is the unique
    winner of $R$-election $(C,V - V')$, where $\|V'\| \leq k$.
  \item In control by adding candidates ($R$-\#CCAC), we ask how many
    sets $A'$, $A' \subseteq A$, are there such that $p$ is the unique
    winner of $R$-election $(C \cup A',V)$, where $\|A'\| \leq k$.
  \item In control by deleting candidates ($R$-\#CCDC), we ask how
    many sets $C'$, $C' \subseteq C$, are there such that $p$ is the
    unique winner of $R$-election $(C - C',V)$, where $\|C'\|
    \leq k$ and $p \not\in C'$.
  \end{enumerate}
  Destructive variants are defined identically, except that we ask for
  the number of settings where the designated candidate is not the
  unique winner.

\end{definition}
\noindent
(To obtain decision variants of the above problems, simply change the
question from asking for a particular quantity to asking if that
quantity is greater than zero. We mention that in the literature
researchers often study both the unique-winner model---as we do
here---and the nonunique-winner model, where it suffices to be one of
the winners. For the rules that we study, results for both models
are the same.)\medskip

The above problems are interesting for several reasons. First, as
described in the introduction, we believe that they are useful as
models for various scenarios pertaining to predicting election
winners. Second, they are counting variants of the well-studied
control
problems~\cite{bar-tov-tri:j:control,hem-hem-rot:j:destructive-control}.
Third, they expose intuitive connections between various types of
control. In particular, we have the following results linking the
complexity of destructive variants with that of the constructive ones,
and linking the complexity of the ``deleting'' variants with that of
the ``adding'' variants.

\begin{theorem}\label{thm:c-to-d}
  Let $R$ be a voting system, let $\#\calC$ be one of $R$-\#CCAC,
  $R$-\#CCDC, $R$-\#CCAV, $R$-\#CCDV, and let $\#\calD$ be the
  destructive variant of $\#\calC$. Then, $\#\calC$ metric reduces to
  $\#\calD$ and $\#\calD$ metric reduces to $\#\calC$.
\end{theorem}
\begin{proof}
  We give a metric reduction from $\#\calC$ to $\#\calD$.  Let $I$ be
  an instance of $\#\calC$, where the goal is to make some candidate
  $p$ the unique winner. We define $f(I)$ to be an instance of
  $\#\calD$ that is identical to $I$, except that the goal is to
  ensure that $p$ is not the unique winner. Let $s_I$ be the total
  number of control actions allowed in $I$\footnote{For example, if
    $\#\calC$ was $\#CCDV$ and $I = (C,V,p,k)$, then $s_I$ would be
    the number of up-to-size-$k$ subsets of $V$ (i.e., the number of
    subsets of voters that can be deleted from $V$).} (and, naturally,
  also the total number of control actions allowed in $f(I)$). It is
  easy to see that $s_I$ is polynomial-time computable and that the
  number of solutions for $I$ is exactly $s_I - \#\calD(f(I))$.  Thus,
  we define $g(I,\#\calD(f(I))) = s_I - \#\calD(f(I))$.  We see that
  the reduction is polynomial-time and correct.
  The same argument shows that
  $\#\calD$ metric reduces to $\#\calC$.
\end{proof}

\begin{theorem}\label{thm:d-to-a}
  Let $R$ be a voting system, then $R$-\#CCDV Turing reduces to
  $R$-\#CCAV and $R$-\#CCDC Turing reduces to $R$-\#CCAC.
\end{theorem}
\begin{proof}
  Let us fix a voting system $R$. The proofs that $R$-\#CCDV Turing
  reduces to $R$-\#CCAV and that $R$-\#CCDC Turing-reduces to
  $R$-\#CCAV are very similar and thus we discuss them jointly.
  Let $I = (C,V,p,k)$ be an input instance of $R$-\#CCDV (of
  $R$-CCDC), where $C$ is the candidate set, $V$ is the collection of
  voters, $p$ is the designated candidate, and $k$ is the upper bound
  on the number of voters that can be removed. We define the following
  transformations of $I$:
  \begin{enumerate}\fixenumerate
  \item For the case of $R$-\#CCDV, for each nonnegative integer $g$,
    $g \le \|V\|$, let $J_g$ be an instance of $R$-\#CCAV, where the
    candidate set is $C$, the set of registered voters is empty, the
    collection of unregistered voters is $V$, the designated candidate
    is $p$, and the bound on the number of voters that can be added is
    $g$, That is, $J_g = (C,\emptyset,V,p,g)$.
 
  \item For the case of $R$-\#CCDC, for each nonnegative integer $g$,
    $g \le \|C\|-1$, let $J_g$ be an instance of $R$-\#CCAC, where the
    candidate set is $\{p\}$, the set of registered voters is $V$, the
    set of unregistered candidates is $C-\{p\}$, the designated
    candidate is $p$, and the bound on the number of candidates that
    can be added is $g$. That is, $J_g = (\{p\},C-\{p\},V,p,g)$.
  \end{enumerate}
  To complete the reduction, it suffices to note that for
  the case of $R$-\#CCDV it holds that:
  \[
    \text{$R$-\#CCDV}(I) = \left\{
    \begin{array}{ll}
      \text{$R$-\#CCAV}(J_{\|V\|}) & \text{, if $k \ge \|V\|$}\\
      \text{$R$-\#CCAV}(J_{\|V\|}) - \text{$R$-\#CCAV}(J_{\|V\|-k-1}) & \text{, if $0 \leq k < \|V\|$}\\
    \end{array}
    \right.
  \]
  and for the case of $R$-\#CCDC it holds that:
  \[
    \text{$R$-\#CCDC}(I) = \left\{
    \begin{array}{ll}
      \text{$R$-\#CCAC}(J_{\|C\|-1}) & \text{, if $k \ge \|C\|-1$}\\
      \text{$R$-\#CCAC}(J_{\|C\|-1}) - \text{$R$-\#CCAC}(J_{\|C\|-k-2}) & \text{, if $0 \leq k < \|C\|-1$}\\
    \end{array}
    \right.
  \]
  These expressions define a natural algorithm for solving $R$-\#CCDV
  ($R$-\#CCDC) given at most two calls to an $R$-\#CCAV (an
  $R$-\#CCAC) oracle.  It is clear that this Turing reduction runs in
  polynomial time.
\end{proof}

Theorems~\ref{thm:c-to-d} and~\ref{thm:d-to-a} are very useful and, in
particular, we obtain all of our destructive-case hardness results via
Theorem~\ref{thm:c-to-d}, and all our ``deleting''-case easiness
results using Theorem~\ref{thm:d-to-a}.

Obtaining results similar to Theorems~\ref{thm:c-to-d}
and~\ref{thm:d-to-a} for the decision variants of control problems is
more difficult. Nonetheless, recently several researchers have made
some progress on this front. In particular, Hemaspaandra,
Hemaspaandra, and Menton~\cite{hem-hem-men:c:search-decision} have
shown that control by partition of candidates and control by runoff
partition of candidates (two types of control not discussed in this
paper) are equivalent in terms of their computational complexity.
Focusing on particular classes of voting rules, for the case of
$k$-Approval and $k$-Veto, Faliszewski, Hemaspaandra, and
Hemaspaandra~\cite{fal-hem-hem:c:weighted-control} have shown that
control by deleting voters reduces to control by adding voters, and
Mi\k{a}sko~\cite{mia:t:priced-control} achieved the same for the case
of voting rules based on the weighted majority graphs (e.g., for
Borda, Condorcet, Maximin, and Copeland).\footnote{Mi\k{a}sko's
  master's thesis gives this result for the case of Borda and
  destructive control only, but it is clear that the proof technique
  adapts to constructive control and that it applies to all voting
  rules based on weighted majority graphs.}

\section{Results}
\label{sec:results}

We now present our complexity results regarding counting variants of
election control problems. We consider two cases: the unrestricted
case, where the voters can have any arbitrary preference orders, and
the single-peaked case, where voters' preference orders are
single-peaked with respect to a given societal axis. We summarize our
results in Table~\ref{tab:results}. For comparison, in
Table~\ref{tab:decision} we quote results for the decision variants of
the respective control problems.

Let us consider the unrestricted case first.  Not surprisingly, for
each of our control problems whose decision variant is $\np$-complete,
the counting variant is $\sharpp$-complete for some reducibility type.
However, interestingly, we have also found examples of election
systems and control types where the decision variant is easy, but the
counting variant is hard. For example, for $2$-Approval we have that
all decision variants of voter control problems are in
$\p$~\cite{lin:thesis:elections,fal-hem-hem:c:weighted-control}, yet
all their counting variants are $\sharpp$-Turing-complete. Similarly,
for Maximin all decision variants of candidate control (except
constructive control by adding candidates) are in
$\p$~\cite{fal-hem-hem:j:multimode}, yet their counting variants are
$\sharpp$-Turing-complete (or are complete for $\sharpp$ through even
less demanding reducibility types).

The situation for the single-peaked case is quite interesting as well.
While we found polynomial-time algorithms for all of our control
problems for Plurality, $k$-Approval, and all Condorcet-consistent
rules, we did not obtain results for voter control under Approval
(candidate control for Approval is easy even in the unrestricted
case).  Decision variants of voter control are $\np$-complete for
Approval in the unrestricted case, but---very interestingly---are in
$\p$ for the single-peaked case. However, the algorithm for the
single-peaked case is a clever greedy approach that seems to be very
difficult to adapt to the counting case. Further, as opposed to the
$k$-Approval case, for voter control under Approval, there does not
seem to be a natural dynamic-programming-based approach.  Naturally,
we also tried to find a hardness proof, but we failed at that as well
because the problem seems to have quite a rich structure.  (However,
see the work of Mi\k{a}sko~\cite{mia:t:priced-control} for the case of
priced control under Approval with single-peaked preferences.)

\newcommand{\complete}{com.}

\tabcolsep=0.1cm
\begin{table*}[t]
\small

\begin{center}

(a) The Unrestricted Case

\hfill{}
	\begin{tabular}{c|cccccc}
          \hline
          Problem\rule{0cm}{0.5cm} & Plurality & $k$-Approval, $k \geq 2$ & Approval & Condorcet  & Maximin\\[0.15cm]
          \hline 
          \#CCAC\rule{0cm}{0.5cm}& $\sharpp$-\complete & $\sharpp$-\complete &-- & --  & $\sharpp$-\complete\\
          \#DCAC & $\sharpp$-metric-\complete & $\sharpp$-metric-\complete & FP & FP  & $\sharpp$-metric-\complete \\
          \#CCDC & $\sharpp$-\complete & $\sharpp$-\complete &FP & FP & $\sharpp$-Turing-\complete \\
          \#DCDC & $\sharpp$-metric-\complete & $\sharpp$-metric-\complete & -- & -- & $\sharpp$-Turing-\complete \\[0.4em]
          \#CCAV & $\fp$ & $\sharpp$-Turing-\complete & $\sharpp$-\complete & $\sharpp$-\complete & $\sharpp$-\complete \\
          \#DCAV & $\fp$ & $\sharpp$-Turing-\complete & $\sharpp$-metric-\complete & $\sharpp$-metric-\complete & $\sharpp$-metric-\complete  \\
          \#CCDV & $\fp$ & $\sharpp$-metric-\complete & $\sharpp$-\complete & $\sharpp$-\complete & $\sharpp$-\complete \\
          \#DCDV & $\fp$ & $\sharpp$-metric-\complete & $\sharpp$-metric-\complete & $\sharpp$-metric-\complete & $\sharpp$-metric-\complete  \\
	\end{tabular}

\vspace{0.5cm}

(b) The Single-Peaked Case

\hfill{}

	\begin{tabular}{c|cccc}
          \hline
          Problem\rule{0cm}{0.5cm} & Plurality & $k$-Approval, $k \geq 2$ & Approval & Condorcet-Consistent \\
                                   &           &                          &          & (e.g., Maximin)\\[0.15cm]
          \hline 
          \#CCAC\rule{0cm}{0.5cm}& $\fp$ & $\fp$ &-- & --  \\
          \#DCAC & $\fp$ & $\fp$ & FP & FP   \\
          \#CCDC & $\fp$ & $\fp$ &FP & FP  \\
          \#DCDC & $\fp$ & $\fp$ & -- & --  \\[0.4em]
          \#CCAV & $\fp$ & $\fp$ & ? & $\fp$  \\
          \#DCAV & $\fp$ & $\fp$ & ? & $\fp$   \\
          \#CCDV & $\fp$ & $\fp$ & ? & $\fp$  \\
          \#DCDV & $\fp$ & $\fp$ & ? & $\fp$   \\
	\end{tabular}

	\hfill{}
	\caption{\label{tab:results}The complexity of counting
          variants of control problems for (a) the unrestricted case,
          and for (b) the single-peaked case. A dash in an entry means
          that the given system is \emph{immune} to the type of
          control in question (i.e., it is impossible to achieve the
          desired effect by the action this control problem allows;
          technically this means the answer to the counting question
          is always $0$). Immunity results were established by
          Bartholdi, Tovey, and
          Trick~(1989) %
          for the constructive cases, and by Hemaspaandra,
          Hemaspaandra, and
          Rothe~(2007) %
          for the destructive cases.  }

\end{center}
\end{table*}

\tabcolsep=0.1cm
\begin{table*}[t]
\small

\begin{center}

(a) The Unrestricted Case

\hfill{}

	\begin{tabular}{c|ccccccc}
          \hline
          Problem\rule{0cm}{0.5cm} & Plurality & $2$-Approval & $3$-Approval &$k$-Approval, $k \geq 4$ & Approval & Condorcet  & Maximin\\[0.15cm]
          \hline 
          CCAC\rule{0cm}{0.5cm}& $\np$-\complete & $\np$-\complete & $\np$-\complete & $\np$-\complete &-- & --  & $\np$-\complete\\
          DCAC & $\np$-\complete & $\np$-\complete & $\np$-\complete & $\np$-\complete& $\p$ & $\p$  & $\p$ \\
          CCDC & $\np$-\complete & $\np$-\complete & $\np$-\complete & $\np$-\complete& $\p$ & $\p$ & $\p$ \\
          DCDC & $\np$-\complete & $\np$-\complete & $\np$-\complete & $\np$-\complete& -- & -- & $\p$ \\[0.4em]
          CCAV & $\p$ & $\p$  & $\p$ & $\np$-\complete & $\np$-\complete & $\np$-\complete & $\np$-\complete \\
          DCAV & $\p$ & $\p$ & $\p$ & $\np$-\complete & $\p$ & $\p$ & $\np$-\complete  \\
          CCDV & $\p$ & $\p$ & $\np$-\complete & $\np$-\complete & $\np$-\complete & $\np$-\complete & $\np$-\complete \\
          DCDV & $\p$ & $\p$ & $\p$ & $\np$-\complete & $\p$ & $\p$ & $\np$-\complete  \\
	\end{tabular}

\vspace{0.5cm}

(b) The Single-Peaked Case

\hfill{}

	\begin{tabular}{c|cccc}
          \hline
          Problem\rule{0cm}{0.5cm} & Plurality & $k$-Approval, $k \geq 2$ & Approval & Condorcet-Consistent \\
                                   &           &                          &          & (e.g., Maximin)\\[0.15cm]
          \hline 
          CCAC\rule{0cm}{0.5cm}& $\p$ & $\p$ &-- & --  \\
          DCAC & $\p$ & $\p$ & $\p$ & $\p$   \\
          CCDC & $\p$ & $\p$ & $\p$ & $\p$  \\
          DCDC & $\p$ & $\p$ & -- & --  \\[0.4em]
          CCAV & $\p$ & $\p$ & $\p$ & $\p$  \\
          DCAV & $\p$ & $\p$ & $\p$ & $\p$   \\
          CCDV & $\p$ & $\p$ & $\p$ & $\p$  \\
          DCDV & $\p$ & $\p$ & $\p$ & $\p$   \\
	\end{tabular}

	\hfill{}
	\caption{\label{tab:decision}The complexity of decision
          variants of control problems for (a) the unrestricted case,
          and for (b) the single-peaked case. A dash in an entry means
          that the given system is \emph{immune} to the type of
          control in question. For the unrestricted case, we have the
          following: The results for Plurality are due to Bartholdi,
          Tovey, and Trick~\cite{bar-tov-tri:j:control} and
          Hemaspaandra, Hemaspaandra, and
          Rothe~\cite{hem-hem-rot:j:destructive-control}, the results
          regarding $k$-Approval are due to
          Lin~\cite{lin:thesis:elections} (see
          also~\cite{elk-fal-sli:j:cloning,fal-hem-hem:c:weighted-control}),
          the results regarding Approval are due to Hemaspaandra,
          Hemaspaandra, and
          Rothe~\cite{hem-hem-rot:j:destructive-control}, and the
          results regarding Maximin are due to Faliszewski,
          Hemaspaandra, and
          Hemaspaandra~\cite{fal-hem-hem:j:multimode}.  For the
          single-peaked case, we inherit polynomial-time algorithms
          from the unrestricted case. The remaining results (except
          those for Condorcet-consistent rules) are due to Faliszewski
          et al.~\cite{fal-hem-hem-rot:j:single-peaked-preferences},
          and the remaining results regarding Condorcet-consistent
          rules are due to~\cite{bra-bri-hem-hem:c:sp2}.}

\end{center}
\end{table*}

In the following sections we give proofs for our results and provide
some more detailed discussion regarding particular voting rules.

\subsection{Plurality Voting}
Under plurality voting, counting variants of both control by adding
voters and control by deleting voters are in $\fp$ even for the
unrestricted case. Our algorithms are based on dynamic programming and
applications of Theorems~\ref{thm:c-to-d} and~\ref{thm:d-to-a}.

\begin{theorem} \label{th:plavc} 
  Plurality-\#CCAV, Plurality-\#DCAV, Plurality-\#CCDV, and
  Plurality-\#DCDV are in $\fp$, even in the unrestricted case.
\end{theorem}

\begin{proof}
  We give the proof for Plurality-\#CCAV only. The result for
  Plurality-\#CCDV follows by applying Theorem~\ref{thm:d-to-a}, and
  the destructive cases follows by applying Theorem~\ref{thm:c-to-d}.

  Let $I = (C,V,W,p,k)$ be an input instance of Plurality-\#CCAV,
  where $C = \{p,c_1,\ldots, c_{m-1}\}$ is the candidate set, $V$ is
  the set of registered voters, $W$ is the set of unregistered voters,
  $p$ is the designated candidate, and $k$ is the upper bound on the
  number of voters that can be added. We now describe a polynomial-time
  algorithm that computes the number of solutions for $I$.

  Let $A_p$ be the set of voters from $W$ that rank $p$ first.
  Similarly, for each $c_i \in C$, let $A_{c_i}$ be the set of voters
  from $W$ that rank $c_i$ first. We also define
  $\id{count}(C,V,W,p,k,j)$ to be the number of sets $W' \subseteq W - A_p$
  such that
    (1) $\|W'\| \leq k-j$, and
    (2) in election $(C,V \cup W')$ each candidate $c_i \in C$, $1
    \leq i \leq m-1$, has score at most $\score{(C,V)}^p(p)+j-1$.

  Our algorithm works as follows.  First, we compute $k_0$, the
  minimum number of voters from $A_p$ that need to be added to $V$ to
  ensure that $p$ has plurality score higher than any other candidate
  (provided no other voters are added). Clearly, if $p$ already is the
  unique winner of $(C,V)$ then $k_0$ is $0$, and otherwise $k_0$ is
  $\max_{c_i \in C}(\score{(C,V)}^p(c_i)-\score{(C,V)}^p(p)+1)$. Then,
  for each $j$, $k_0 \leq j \leq \min(k,\|A_p\|)$, we compute the number
  of sets $W'$, $W' \subseteq W$, such that $W'$ contains exactly $j$
  voters from $A_p$, at most $k-j$ voters from $W - A_p$, and $p$ is
  the unique winner of $(C, V \cup W')$. It is easy to verify that for
  a given $j$, there is exactly $h(j) =
  \binom{\|A_p\|}{j}\cdot\id{count}(C,V,W\!,p,k,j)$ such sets. Our
  algorithm returns $\sum_{j=k_0}^{\min(k,\|A_p\|)}h(j)$. The reader can
  easily verify that this indeed is the correct answer. To complete
  the proof it suffices to show a polynomial-time algorithm for
  computing $\id{count}(C,V,W,p,k,j)$.

  Let us fix $j$, $k_0 \leq j \leq \min(k,\|A_p\|)$. We now show how to
  compute $\id{count}(C,V,W,p,k,j)$. Our goal is to count the number
  of ways in which we can add at most $k-j$ voters from $W-A_p$ so
  that no candidate $c_i \in C$ has score higher than
  $\score{(C,V)}^p(p)+j-1$.  For each candidate $c_i \in C$, we can
  add at most
  $
  l_i =
  \min\bigl(\|A_{c_i}\|,j+\score{(C,V)}^p(p)-\score{(C,V)}^p(c_i)-1\bigr)
  $
  voters from $A_{c_i}$; otherwise $c_i$'s score would exceed
  $\score{(C,V)}^p(p)+j-1$.

  For each $i$, $1 \leq i \leq m-1$, and each $t$, $0 \leq t \leq
  k-j$, let $a_{t,i}$ be the number of sets $W' \subseteq A_{c_1}\cup
  A_{c_2}\cup\cdots\cup A_{c_i}$ that contain exactly $t$ voters and
  such that each candidate $c_1, c_2, \ldots, c_i$ has score at most
  $\score{(C,V)}^p(p)+j-1$ in the election $(C,V \cup W')$. Naturally,
  $\id{count}(C,V,W,p,k,j) = \sum_{t=0}^{k-j}a_{t,m-1}$.  It is easy
  to check that $a_{t,i}$ satisfies the following recursion:
  \[
  a_{t,i} =
  \begin{cases}
    \sum_{s=0}^{\min(l_i,t)}\binom{\|A_{c_i}\|}{s}a_{t-s,i-1}, & \text{if $t>0$, $i>1$}, \\[2mm]
    1, & \text{if $t=0$, $i>1$}, \\[1mm]
    \binom{\|A_1\|}{t}, & \text{if $t\le\|A_{c_1}\|$, $i=1$}, \\[1mm]
    0, & \text{if $t>\|A_{c_1}\|$, $i=1$}.
  \end{cases}
  \]
  Thus, for each $t,i$ we can compute $a_{t,i}$ using standard dynamic
  programming techniques in polynomial time. Thus,
  $\id{count}(C,V,W,p,k,j)$ also is polynomial-time computable.  This
  completes the proof that Plurality-\#CCAV is in $\fp$.
\end{proof}

On the other hand, for Plurality voting \#CCAC and \#CCDC are
$\sharpp$-complete and this follows from proofs already given in the
literature~\cite{fal-hem-hem:j:nearly-sp}.

\begin{theorem}\label{thm:placdcc}
  In the unrestricted case, Plurality-\#CCAC and Plurality-\#CCDC are $\sharpp$-complete.
\end{theorem}
\begin{proof}
  Faliszewski, Hemaspaandra, and
  Hemaspaandra~\cite[Theorem~6.4]{fal-hem-hem:j:nearly-sp} show that
  decision variants of control by adding candidates and control by
  deleting candidates are $\np$-complete.\footnote{Naturally, these
    results are originally due to Bartholid, Tovey, and
    Trick~\cite{bar-tov-tri:j:control}. However, the proofs of Faliszewski, Hemaspaandra,
    and Hemaspaandra are more useful in our case, because they work
    for $k$-Approval for each fixed $k$ (which will be useful later),
    and it is convenient to verify that they give parsimonious
    reductions from a $\sharpp$-complete problem.} Their proofs work
  by reducing X3C to appropriate control problems in a parsimonious
  way. This means that the same reductions reduce \#X3C to the
  counting variants of respective control problems.
\end{proof}

\noindent
Now, Corollary~\ref{cor:placdcdc} follows by combining
Theorems~\ref{thm:placdcc} and~\ref{thm:c-to-d}.

\begin{corollary}\label{cor:placdcdc}
  In the unrestricted case, Plurality-\#DCAC and Plurality-\#DCDC are $\sharpp$-metric-complete.
\end{corollary}

However, for the  single-peaked case, we get easiness results. 
\begin{theorem}
  For the case of single-peaked profiles, Plurality-\#CCAC,
  Plurality-\#CCDC, Plurality-\#DCAC, and Plurality-\#DCDC are in
  $\fp$.
\end{theorem}
\begin{proof}
  This result follows from Theorem~\ref{thm:kapp-ccac-sp} (see the next section) for
  the case of $1$-Approval.
\end{proof}

\subsection{$\boldsymbol{k}$-Approval Voting}
While $k$-Approval is in many respects a simple generalization of the
Plurality rule, it turns out that for $k \geq 2$, for the
unrestricted case, all counting variants of control
problems are intractable for $k$-Approval.  This is quite expected for
candidate control because decision variants of these problems are
$\np$-complete (see the work of
Lin~\cite{lin:thesis:elections,elk-fal-sli:j:cloning} and additionally
of Faliszewski, Hemaspaandra, and
Hemaspaandra~\cite{fal-hem-hem:c:weighted-control}), but is more
intriguing for voter control (as shown by Lin, for $2$-Approval all
voter control decision problems are in $\p$, and, as one can verify,
for $k$-Approval all destructive voter control decision problems are
in $\p$).
On the other hand, for the single-peaked case all the counting
variants of control are polynomial-time solvable for $k$-Approval
(however, note that we rely on $k$ being a constant; our results do
not generalize to the case where $k$ is part of the input).  \bigskip

\noindent
\textbf{Candidate Control.}
\quad
We start be quickly dealing with candidate control in the unrestricted
case.  For the unrestricted case, as in the case of Plurality, we can
use the result of Faliszewski, Hemaspaandra, and
Hemaspaandra~\cite{fal-hem-hem:j:nearly-sp}.

\begin{theorem}
  In the unrestricted case, for each $k$, $k \geq 2$,
  $k$-Approval-\#CCAC and $k$-Approval-\#CCDC are $\sharpp$-complete,
  and $k$-Approval-\#DCAC and $k$-Approval-\#DCDC are
  $\sharpp$-metric-complete.
\end{theorem}
\begin{proof}
  For the constructive variants, the same approach as in the proof of
  Theorem~\ref{thm:placdcc} works: The hardness proofs given by
  Faliszewski, Hemaspaandra, and
  Hemaspaandra~\cite[Theorem~6.4]{fal-hem-hem:j:nearly-sp} apply to
  the case of $k$-Approval as well.  For the destructive variants, we
  invoke Theorem~\ref{thm:c-to-d}.
\end{proof}

For the single-peaked case, we use a dynamic-programming approach. Our
algorithm is a very extensive generalization of the algorithm for the
single-peaked variant of Plurality-CCAC given by Faliszewski et
al.~\cite{fal-hem-hem-rot:j:single-peaked-preferences}

\begin{theorem}\label{thm:kapp-ccac-sp}
  For the case of single-peaked profiles, for each fixed positive $k$,
  $k$-Approval-\#CCAC, $k$-Approval-\#CCDC, $k$-Approval-\#DCAC, and
  $k$-Approval-\#DCDC are in $\fp$.
\end{theorem}
\begin{proof}
  We give the proof for $k$-Approval-\#CCAC only. The result for
  $k$-Approval-\#CCDC follows by applying Theorem~\ref{thm:d-to-a},
  and the destructive cases follow by applying
  Theorem~\ref{thm:c-to-d}.

  We are given a candidate set $C$, a voter collection $V$, a
  designated candidate $p \in C$, a collection $A$ of unregistered
  candidates, a positive integer $\ell$, and a societal axis
  $\mathrel{L}$ over $C \cup A$.  We show a polynomial time algorithm
  for determining the number of sets $A'$, $A' \subseteq A$, where
  $\|A'\| \le \ell$, such that $p$ is the unique $k$-Approval winner
  in election $(C \cup A', V)$. We assume that $\|C \cup A\| > 4k$
  (otherwise we solve our problem by enumerating all possible solutions).

  Intuitively, the idea of our algorithm is as follows: First, we fix
  two groups of $k$ candidates, one ``to the left'' of $p$ (with
  respect to $L$) and one ``to the right'' of $p$.  We will argue that
  each choice of such ``neighborhood'' of $p$ fixes $p$'s score.
  Second, with $p$'s score fixed, we count how many ways are there to
  add candidates so that every candidate has fewer points than $p$.
  We sum these values over all choices of $p$'s ``neighborhood.''

  \newcommand{\before}{{{\mathrm{before}}}}
  \newcommand{\after}{{{\mathrm{after}}}}

  Let us introduce some useful notation. For each set $X \subseteq C
  \cup A$ of candidates we define $\before(X)$ to be the subset of
  those candidates in $C \cup A$ that precede each member of $X$ with
  respect to $L$ (i.e., $\before(X) = \{ d \in C \cup A \mid (\forall
  x \in X)[d \spl x]\}$). Analogously, we define $\after(X)$ to be the
  subset of those candidates in $C \cup A$ that succeed all members of
  $X$ with respect to $L$ (i.e., $\after(X) = \{ d \in C \cup A \mid
  (\forall x \in X)[x \spl d]\}$). The next lemma describes how we can
  fix candidates' scores by fixing their ``neighborhoods.''

  \begin{lemma}
    \label{k-barrier-lemma}
    For each set $H \subseteq C \cup A$, $\|H\| = k$, each set $X
    \subseteq \before(H)$ and each candidate $c \in X$, if $z$ is the score
    of candidate $c$ in k-Approval election $(X \cup H, V)$, then for
    each $Y \subseteq \after(H)$ in k-Approval election $(X \cup H \cup Y,
    V)$ the score of candidate $c$ is $z$ as well.
  \end{lemma}

  \begin{proof}
    For a given subset $Y \subseteq \after(H)$ and candidate $c \in
    X$, suppose that the score of candidate $c$ in $k$-Approval election
    $(X \cup H \cup Y, V)$ is $z' \ne z$.  Adding a candidate can
    never increase the score of a candidate already present, so we
    conclude that $z' < z$.  It means that there is a voter $v \in V$
    from which $c$ gains a point in election $E_1 = (X \cup H,V)$ but
    not in election $E_2 = (X \cup H \cup Y, V)$.  Thus there is a
    candidate $c' \in Y$ that receives a point from $v$.  Since $c$ is
    among $v$'s top $k$ most preferred candidates across $X \cup H$
    and $\|H\| = k$, there must exist a candidate $c'' \in H$ that is
    less preferred than $c$ by $v$.  Since $v$ prefers $c'$ over $c$,
    it must be that $v$ prefers $c'$ over $c''$. However, we know that
    $c \mathrel{L} c'' \mathrel{L} c'$, because $c \in X$, $c'' \in H$
    and $c' \in Y$.  Since $v$ is single-peaked with respect to 
    $L$, $v$ cannot prefer both $c$ and $c'$ over $c''$; a contradiction.
  \end{proof}

  Let $B_L$ be the set of the first $2k$ candidates from $C \cup A$
  (with respect to $L$) and let $B_R$ be the set of the last $2k$
  candidates from $C \cup A$ (with respect to $L$).  Without loss of
  generality, we assume that $B_L$ and $B_R$ contain candidates from
  $C$ only, and that each voter has a preference order of the form $C
  - (B_L \cup B_R) > B_L > B_R$ (if this were not the case then we
  could add to $C$ two groups of $2k$ dummy candidates that all the
  voters rank last, respecting the above requirement, and that are at
  the extreme ends of the societal axis $L$).  Note that by our
  assumptions $p$ does not belong to $B_L \cup B_R$ and for each $A''
  \subseteq A$ it holds that in election $(C \cup A'', V)$ the
  candidates from $B_R$ receive $0$ points each.

  We say that a set $Y \subseteq C \cup A$ is well-formed if for each
  candidate $d \in C$ it holds that $d \in \before(Y) \cup Y \cup
  \after(Y)$ (in other words, $Y$ is well-formed if it is an interval
  with respect to the societal axis $L$ restricted to $C \cup Y$).  We
  define $K_L$ to be a collection of subsets of $\before(\{p\})$ such
  that if $Y \in K_L$ then $Y\cup\{p\}$ is well-formed and $\|Y\| = k$.
  (Intuitively, $K_L$ is the family of sets of candidates that can
  form the ``left part'' of $p$'s neighborhood.)  We define $K_R$
  analogously, but ``to the right of $p$.'' That is, $K_R$ is a
  collection of subsets of $\after(\{p\})$ such that if $Y \in K_R$
  then $\{p\} \cup Y$ is well-formed and $\|Y\| = k$.
  Note that there are at most $O(\|C \cup A\|^k)$ sets in each of
  $K_L$ and $K_R$.

  Our algorithm proceeds as follows.  In a loop, we try each $Y_L \in
  K_L$ and each $Y_R \in K_R$. For each choice of $Y_L$ and $Y_R$, we
  set $Y = Y_L \cup \{p\} \cup Y_R$ to be the neighborhood of $p$, and
  we set $\rho$ to be the $k$-Approval score of $p$ in election
  $(Y,V)$. (If $\|Y \cap A\| > \ell$ then we drop this choice of $Y_L$
  and $Y_R$ because picking this neighborhood requires adding more
  candidates than we are allowed to.) We create two new sets, $C' = C
  \cup Y$, $A' = A \cap (\before(Y) \cup \after(Y))$, and an integer
  $\ell' = \ell - \|Y \cap A\|$. Finally, we compute the number $s(Y)$
  of size-at-most-$\ell'$ subsets $A''$ of $A'$ such that in
  election $(C' \cup A'',V)$ each candidate has less than $\rho$
  points (except for $p$ who, by Lemma~\ref{k-barrier-lemma}, has
  exactly $\rho$ points). Computing $s(Y)$ in polynomial time is a
  crucial technical part of the algorithm and we describe it below. We
  sum up all the values $s(Y)$ and return them as our output. By
  Lemma~\ref{k-barrier-lemma} it is easy to see that this strategy is
  correct.

  We now describe how to compute $s(Y)$ in polynomial time.  Let us
  consider a situation where we have picked $p$'s neighborhood $Y$ and
  computed $\rho$, $C'$, $A'$, and $\ell'$. If $\rho = 0$ then $p$
  cannot be the unique winner so we assume that $\rho > 0$.
  From now on, we assume that $C'$ and $A'$ take the roles of $C$ and
  $A$ in the definition of a well-formed set.  

  \newcommand{\prev}{{{\mathrm{prev}}}}
  \newcommand{\Prev}{{{\mathrm{Prev}}}}

  For each well-formed set $Z = \{z_1, \ldots, z_{2k}\} \subseteq C'
  \cup A'$, such that $z_1 \spl \cdots \spl z_{2k}$, and for each
  nonnegative integer $s$ we define:
  \begin{enumerate}
  \item[] $f(Z,s) = $ the number of sets $A'' \subseteq (A' \cap
    \before(Z))$ such that $\|A''\| \leq s$ and in election $(Z \cup
    A'' \cup (C' \cap \before(Z)))$ each candidate in $\{z_1, \ldots,
    z_k\} \cup A'' \cup (C' \cap \before(Z))$ other than $p$ (if $p$
    is included in this set) has fewer than $\rho$
    points.\footnote{Yes, we really mean the condition on scores to
      apply to candidates $z_1, \ldots, z_k$ but not to candidates
      $z_{k+1}, \ldots, z_{2k}$; we will implicitly ensure that the
      scores of $z_{k+1}, \ldots, z_{2k}$ do satisfy the condition as
      well.}
  \end{enumerate}
  It is easy to see that $s(Y)$ is simply $f(B_R,\ell')$ (recall that
  $B_R$ is the set of ``right-hand side'' dummy candidates, who have score
  $0$ in every election). For each $Z$ and $s$, we can compute
  $f(Z,s)$ using dynamic programming. To provide appropriate recursive
  expression for $f$, we need some additional notation.  We define
  $\prev(Z)$ to be the last candidate from $C'$ with respect to $L$
  that precedes the candidates from $Z$. That is $\prev(Z)$ is the
  maximal (``rightmost'') element of $C' \cap \before(Z)$ with respect
  to $L$.  We define $\Prev(Z)$ to be the set that contains $\prev(Z)$
  and all the candidates from $C' \cup A'$ that are between $\prev(Z)$
  and $Z$, with respect to $L$. For a well-formed set $Z = \{z_1,
  \ldots, z_{2k}\}$ and an element $z \in \Prev(Z)$ (provided that
  $\Prev(Z)$ is defined) we let $\delta(Z,z)$ be $1$ if in election
  $(\{z,z_1,\ldots, z_{2k}\},V)$ candidate $z_k$ is either $p$ or
  obtains fewer than $\rho$ points. Otherwise we set $\delta(Z,z) =
  0$.  Now, for each nonnegative integer $s$ and each well-formed set
  $Z = \{z_1, \ldots, z_{2k}\} \subseteq C' \cup A'$ such that $z_1
  \spl \cdots \spl z_{2k}$ the following relation holds ($\chi_{A'}$
  is the characteristic function of set $A'$, i.e., $\chi_{A'}(z)$ is
  $1$ if $z \in A'$ and is $0$ otherwise):
  \[ f(Z,s) = \left\{ \def\arraystretch{1.5}
    \begin{array}{l}
      \displaystyle\sum_{z \in \Prev(Z)} \delta(Z,z)f(\{z,z_1,\ldots,z_{2k-1}\}, s-\chi_{A'}(z))  \text{, if $\Prev(Z)$ is defined,} \\
      \quad 1  \text{, otherwise.}
    \end{array} \right.
  \]
  Note that if $\Prev(Z)$ is not defined then $Z = B_L$ and by
  single-peakedness of the election, $z_1, \ldots, z_k$ have zero
  points each.

  The correctness of this recursive expression follows by our
  assumptions and by Lemma~\ref{k-barrier-lemma}. It is easy to see
  that using dynamic programming and this recursive expression we can
  compute $f(B_R,\ell')$ in polynomial time. This concludes the proof.
\end{proof}

\noindent
\textbf{Voter Control.}  \quad We now turn to the case of voter
control and we start with the unrestricted case. This time, we need
quite a few new ideas: Under $k$-Approval (for fixed $k$, $k \geq 2$)
all types of counting voter control are hard, while the complexity of
decision variants is quite varied (see the works of
Lin~\cite{lin:thesis:elections,elk-fal-sli:j:cloning} and of Faliszewski, Hemaspaandra, and
Hemaspaandra~\cite{fal-hem-hem:c:weighted-control}).

We start by considering $2$-Approval-\#CCAV and $2$-Approval-\#CCDV
separately, then we extend these results to $k \geq 3$, and finally we
invoke Theorem~\ref{thm:c-to-d} to obtain the results for the
destructive cases.

\begin{theorem}\label{thm:ccavdv-2}
  In the unrestricted case,
  $2$-Approval-\#CCAV is $\sharpp$-Turing-complete and
  $2$-Approval-\#CCDV is $\sharpp$-metric-complete.
\end{theorem}
\begin{proof}
  We first consider $2$-Approval-\#CCAV.
  We give a Turing reduction from \#PerfectMatching to
  $2$-Approval-\#CCAV.  Let $G = (G(X),G(Y),G(E))$ be our input
  bipartite graph, where $G(X) = \{x_1, \ldots, x_n\}$ and $G(Y) =
  \{y_1, \ldots, y_n\}$ are sets of vertices, and $G(E) = \{e_1,
  \ldots, e_m\}$ is the set of edges. We form an election $E = (C,V)$
  and a collection $W$ of unregistered voters as follows. We set $C =
  \{p,b_1,b_2\} \cup G(X) \cup G(Y)$ and we let $V = (v_1,v_2)$, where
  $v_{1}$ has preference order $p > b_1 > C - \{p,b_1\}$ and $v_{2}$
  has preference order $p > b_2 > C - \{p,b_2\}$. We let $W = (w_1,
  \ldots, w_m)$, where for each $\ell$, $1 \leq \ell \leq m$, if
  $e_\ell = \{x_i,y_j\}$ then $w_\ell$ has preference order $x_i > y_j
  > C - \{x_i,y_j\}$.

  Thus, in election $E$ candidate $p$ has score $2$, candidates $b_1$
  and $b_2$ have score $1$, and candidates in $G(X) \cup G(Y)$ have score
  $0$. We form an instance $I$ of $2$-Approval-\#CCAV with election $E
  = (C,V)$, collection $W$ of unregistered voters, designated
  candidate $p$, and the number of voters that can be added set to
  $n$. We form instance $I'$ to be identical, except we allow to add
  at most $n-1$ voters.

  It is easy to verify that the number of $2$-Approval-\#CCAV
  solutions for $I$ (for $I'$) is the number of matchings in $G$ of
  cardinality at most $n$ (the number of matchings in $G$ of
  cardinality at most $n-1$). (Each unregistered voter corresponds to an
  edge in $G$ and one cannot add two edges that share a vertex as then
  $p$ would no longer be the unique winner.) The number of
  perfect matchings in $G$ is exactly the number of solutions for $I$ minus
  the number of solutions for $I'$.\bigskip

  Let us now consider the case of $2$-Approval-\#CCDV. We give a
  metric reduction from \#PerfectMatching. As before, let $G =
  (G(X),G(Y),G(E))$ be our input bipartite graph, where $X = \{x_1,
  \ldots, x_n\}$ and $Y = \{y_1, \ldots, y_n\}$ are sets of vertices,
  and $E = \{e_1, \ldots, e_m\}$ is the set of edges.  For each vertex
  $v$ of $G$, we write $d(v)$ to denote $v$'s degree.  We set $D =
  \max\{d(v) \mid v \in X \cup Y\}$ and we set $T = \sum_{v \in X \cup
    Y}(D-d(v))$. W.l.o.g. we assume that $D \geq 2$.  We form an
  election $E = (C,V)$ as follows.  We set $C = \{p\} \cup X \cup Y
  \cup B$, where $B = \{b_1, \ldots, b_{T+D}\}$.  We form the
  collection $V = V_E + A_{X,Y} + A_p$ of voters as follows:
  \begin{enumerate}\fixenumerate
  \item We set $V_E = (v_1, \ldots, v_m)$ and for each $e_\ell =
    \{x_i, y_j\} \in E$, we set the preference order of $v_\ell$ to be
    $x_i > y_j > C - \{x_i,y_j\}$.
  \item We set $A_{X,Y} = (a_1, \ldots, a_T)$ and we set their
    preference orders so that for each $v \in X \cup Y$, $A_{X,Y}$
    contains exactly $D - d(v)$ voters with preference orders of the
    form $v > b_t > C - \{v,b_t\}$ (see Item \ref{chooseb}. below regarding how
    candidates $b_t$ are chosen).
  \item We set $A_p = \{a_{T+1}, \ldots, a_{T+D}\}$, where each voter
    $a_{T+i}$, $1 \leq i \leq D$, has preference order of the form $p
    > b_t > C - \{p,b_t\}$ (see Item \ref{chooseb}. below regarding how
    candidates $b_t$ are chosen).
  \item\label{chooseb} We arrange the preference orders of voters in $A_{X,Y} + A_p$
    so that each candidate $b_t$, $1 \leq t \leq T+D$, receives
    exactly $1$ point.
  \end{enumerate}
  Thus, before deleting any of the voters, each candidate in $\{p\}
  \cup X \cup Y$ has score $D \geq 2$ and each candidate in $B$ has
  score $1$. We form instance $I$ of $2$-Approval-\#CCDV consisting of
  $E = (C,V)$, designated candidate $p$, and with $n$ as the limit on
  the number of voters we are allowed to delete.

  We claim that the number of solutions for $I$ is equal to the number
  of perfect matchings in $G$. Let $M \subseteq E$ be a perfect
  matching in $G$. We form collection $V' = \{v_i \mid e_i \in
  M\}$. Clearly, $\|V'\| \leq n$ and in election $E' = (C,V-V')$
  candidate $p$ is the unique winner ($p$ has $D$ points, candidates
  in $X \cup Y$ have $D-1$ points, and candidates in $B$ have $1$
  point). 
 
  On the other hand, let $V'$ be a subcollection of $V$ such that
  $\|V'\| \leq n$ and $p$ is the unique winner of election $E' =
  (C,V-V')$. Since $p$ is the unique winner of $E'$, it must hold that
  each of the $2n$ candidates in $X \cup Y$ has at most $D-1$ points
  in $E'$.  Thus, since $|V'| \leq n$, it must be the case that in
  fact $\|V'\| = n$ and $V'$ is a subcollection of $V_E$ such that
  each candidate from $X$ appears in the first position of exactly one
  vote in $V'$ and each candidate from $Y$ appears in the second
  position of exactly one vote from $V'$. As a result, $E' = \{e_i
  \mid v_i \in V'\}$ is a perfect matching for $G$.
\end{proof}

\begin{theorem}
  In the unrestricted case, for each $k \geq 3$, $k$-Approval-\#CCAV
  and $k$-Approval-\#CCDV are $\sharpp$-metric-complete.
\end{theorem}
\begin{proof}
  The proof for the case of \#CCDV follows by applying a natural
  padding argument in the reduction from the proof of
  Theorem~\ref{thm:ccavdv-2}.  Thus we focus on the case of \#CCAV.

  We first give a proof for the case $k=3$.
  We give a Turing reduction from \#PerfectMatching to
  $3$-Approval-\#CCAV.  Let $G = (G(X),G(Y),G(E))$ be our input
  bipartite graph, where $G(X) = \{x_1, \ldots, x_n\}$ and $G(Y) =
  \{y_1, \ldots, y_n\}$ are sets of vertices, and $G(E) = \{e_1,
  \ldots, e_m\}$ is the set of edges. We form an election $E = (C,V)$
  and a collection $W$ of unregistered voters as follows. We set $C =
  \{p,d\} \cup B \cup G(X) \cup G(Y)$ (where $B$ is a collection of dummy
  candidates to be specified later) and we set $V = (v_1,\ldots, v_t)$,
  where the preference orders of the voters in $V$ are such that 
  the score of $p$ is $0$, the score of $d$ is $n-1$, the score
  of each candidate in $G(X) \cup G(Y)$ is $n-2$, and the score of
  each dummy candidate in $B$ is either $1$. 
  We set $W = (w_1, \ldots, w_m)$, where for each $\ell$, $1 \leq \ell \leq m$, if
  $e_\ell = \{x_i,y_j\}$ then $w_\ell$ has preference order $p > x_i > y_j
  > C - \{p,x_i,y_j\}$.

  We form an instance $I$ of $3$-Approval-\#CCAV with election $E
  = (C,V)$, collection $W$ of unregistered voters, designated
  candidate $p$, and the number of voters that can be added set to
  $n$. 

  It is easy to verify that the number of $3$-Approval-\#CCAV
  solutions for $I$ is the number of matchings in $G$ of cardinality
  exactly $n$. (One has to add exactly $n$ voters for $p$ to defeat
  $d$; each unregistered voter corresponds to an edge in $G$ and one
  cannot add two edges that share a vertex as then the candidate
  corresponding to that vertex would have score $n$, and $p$ would not
  be able to become the unique winner.)  Thus, the number of perfect
  matchings in $G$ is exactly the number of solutions for $I$. 
  Further, clearly it is possible to implement our reduction in 
  polynomial time.

   The case for $k > 3$ follows by natural padding arguments and 
   extending the dummy-candidates set $B$.
\end{proof}

Finally, we obtain the results for the destructive cases through
Theorem~\ref{thm:c-to-d} (we remark that it applies to
$2$-Approval-\#CCAV as well because Turing reductions are
generalizations of metric reductions; for the same reason in
Table~\ref{tab:results} we report the complexity for
$k$-Approval-\#CCAV and $k$-Approval-\#CCDV, $k \geq 2$, as
``$\sharpp$-Turing-complete,'' even though our results for $k \geq 3$
are slightly more precise).

\begin{corollary}\label{cor:dc-approval}
  In the unrestricted case, $2$-Approval-\#DCAV is
  $\sharpp$-Turing-complete and for each $k$, $k \geq 2$,
  $(k+1)$-Approval-\#DCAV and $k$-Approval-\#DCDV are
  $\sharpp$-metric-complete.
\end{corollary}

Let us now move on to the single-peaked case. Here we show
polynomial-time algorithms for all counting voter control problems
under $k$-Approval (for fixed $k$). Our algorithm is based on dynamic
programming.  We use the following notation in the proof below: For a
given $j$-element integer vector $\scores$ and integer $i$, $1 \le i \le j$,
we write $\scores_i$ to denote the $i$-{th} element of vector
$\scores$ (i.e., $\scores = (\scores_1,...,\scores_j)$).  

\begin{theorem}
  For the case of single-peaked profiles, for each fixed positive
  integer $k$, $k$-Approval-\#CCAV, $k$-Approval-\#CCDV,
  $k$-Approval-\#DCAV and $k$-Approval-\#DCDV are in $\fp$.
\end{theorem}
\begin{proof}
  We give the proof for $k$-Approval-\#CCAV.  The result for
  $k$-Approval-\#CCDV follows by Theorem~\ref{thm:d-to-a} and the
  results for the destructive cases follow by
  Theorem~\ref{thm:c-to-d}.

  We are given a candidate set $C$, a voter collection $V$, a
  designated candidate $p \in C$, a collection $W$ of unregistered
  voters, a positive integer $\ell$, and a societal axis
  $\mathrel{L}$.  We show a polynomial time algorithm for determining
  the number of sets $W'$, $W' \subseteq W$, where $\|W'\| \le \ell$,
  such that $p$ is a k-Approval winner in election $(C, V \cup W')$.

  For each voter $v \in V \cup W$, we define $\best_k(v)$ to be the
  set of $k$ most preferred candidates according to $v$.  Let
  $\mathrel{\widetilde{L'}}$ be a weak linear order over voter set $W
  \cup V$, such that $v_1 \mathrel{\widetilde{L'}} v_2$, $v_1,v_2 \in
  V \cup W$, if and only if there exists a candidate $c \in
  \best_k(v_2)$ such that for each $c' \in \best_k(v_1)$ we have $c'
  \mathrel{L} c$ or $c' = c$.  Let $\mathrel{L'}$ be a strict order
  obtained from $\mathrel{\widetilde{L'}}$ by breaking the ties in an
  arbitrary fashion.  For each voter $v \in V \cup W$ let $X^v = \{w
  \mid w \in V \cup W \land w \mathrel{L'} v\}$ be the subset of
  voters from $V \cup W$ that precede voter $v$ under $\mathrel{L'}$.
  We let $v_l$ to be the last voter from $V$ under $\mathrel{L'}$.
  For a given voter $v \in V \cup W$, integer $z \in \integers$, and
  integer $s \in \integers$, we define $g(v,z,s)$ to be the number of
  sets $\widetilde{W} \subseteq X^v \cap W$, where $\|\widetilde{W}\|
  = s$, such that $p$ is a k-Approval winner in election $(C, (X^v
  \cap V) \cup \widetilde{W} \cup \{v\})$ and in addition the score of
  candidate $p$ in this election is exactly $z$.  Equation
  \eqref{k-app-ccav-main} below gives the result that we should output
  on the given input instance of k-Approval-\#CCAV problem:

  \begin{equation}\label{k-app-ccav-main}
    \sum_{0 \le s \le \ell} \left( \sum_{0 \le z \le \|W \cup V\|} \left( g(v_l,z,s) + \sum_{\substack{w \in W \\ v_l \mathrel{L'} w}} g(w,z,s) \right) \right)
  \end{equation}

  Thus, it suffices to show how to compute $g(w,z,s)$ for $w \in V
  \cup W$ and $z,s \in \integers$ in polynomial time.  Before we give
  an appropriate algorithm, we need to introduce some additional
  notation.

  For each candidate $c \in C$, let $\rank(c)$ be $c$'s rank under
  $\mathrel{L}$ over all candidates from $C$ (so, for example, if $C =
  \{a,b,c\}$ and $a \mathrel{L} b \mathrel{L} c$ then $\rank(a) = 1$,
  $\rank(b) = 2$ and $\rank(c) = 3$). For each voter $v$, let $\rank(v) =
  \max \{ \rank(c) \mid c \in \best_k(v)\}$.
  For each $v \in V \cup W$, let $\mathcal{P}^v= \{w \mid w \in (V
  \cup W) \land w \mathrel{L'} v \land (\nexists t \in V)[ w
  \mathrel{L'} t \mathrel{L'} v]\}$.  In other words, $\mathcal{P}^v$
  consists of the closest voter $u \in V$ that precedes $v$ under
  $\mathrel{L'}$, and of all the voters that are between $u$ and $v$ under
  $\mathrel{L'}$.  When $v$ is the first element from $V$ under
  $\mathrel{L'}$, then $\mathcal{P}^v$ contains all the voters from
  $W$ preceeding $v$ under $\mathrel{L'}$.

  For a $j$-element integer vector $\scores$ and a nonnegative
  integer $r$, let $\cutoff{\scores}{r}$ denote $j$-element vector
  $(\scores_1, \ldots, \scores_r, 0,\ldots, 0)$ (that is, we replace
  the last $j-r$ entries of vector $\scores$ with zeros).
  For each voter $w \in V \cup W$, we define $\approval(w)$ to be the
  $\|C\|$-dimensional $0/1$ vector that for each candidate $c \in C$
  has $1$ at position $\rank(c)$ if and only if $c \in \best_k(w)$.
  (In other words, $\approval(w)$ is the $0/1$ vector that indicates
  which candidates receive points from voter $w$.) Note that, due to
  single-peakedness of the election, for each voter $w$,
  $\approval(w)$ contains exactly a single consecutive block of $k$
  ones.

  We are now ready to proceed with our algorithm for computing
  function $g$.  For a given voter $v \in V \cup W$, given integers
  $z,s \in \integers$, and a $\|C\|$-dimensional integer vector
  $\scores$, we define $f(v,z,s,\scores)$ to be the number of sets
  $\widetilde{W} \subseteq X^v \cap W$ such that $\|\widetilde{W}\| =
  s$ and in election $(C,(X^v \cap V) \cup \widetilde{W} \cup \{v\})$
  the following holds: (1) $p$ scores exactly $z$ points and (2) each
  candidate $c \in C -\{p\}$ scores no more than $\scores_{\rank(c)}$ points.
  For a given integer $r \in \integers$, let $\Gamma^{r}$ be the
  vector $(r,...,r)$ of dimension $\|C\|$.  Clearly, for a given voter
  $v \in V \cup W$ and given integers $z,s \in \integers$ we have:

  \begin{equation}\label{k-app-ccav-convert}
    g(v,z,s) = f(v,z,s,\Gamma^{z-1})
  \end{equation}

  We now show a recursive formula for $f$.  For a given voter $v \in V
  \cup W$, a given vector $\scores$ of (nonnegative) integers of
  dimension $\|C\|$, and two integers $z,s \in \integers$, we have:

  \begin{equation}\label{k-app-ccav-transform}
    f(v,z,s,\scores) = f(v,z,s,\cutoff{\scores}{\rank(v)})
  \end{equation}

  This follows from the fact that in election $(C,X^v \cup \{v\})$
  only candidates with ranks $j \le \rank(v)$ can score a point.  It is
  easy to see that $f(v,z,s,\scores) = 0$ when $z < 0$ or $\scores$
  contains at least one negative entry, because the score is always a
  nonnegative integer.  When $s = 0$ then $f(v,z,s,\scores)$ is $1$ if
  and only if in election $(C,(X^v \cap V) \cup \{v\})$ candidate $p$
  scores $z$ points and each candidate $c \in C$ scores no more than
  $\scores_{\rank(c)}$ points; otherwise $f(v,z,s,\scores)$ is $0$.
  When $s > 0$, we note that each election consistent with the
  condition for $f(v,z,s,\scores)$ contains at least one voter $w$
  from $\mathcal{P}^v$.  Eq.~\eqref{k-app-ccav-sum} below gives
  formula for $f$ in case when $s > 0$;  for each voter $w \in
  \mathcal{P}^v$ we count all the sets $\widetilde{W} \subseteq X^v
  \cap W$ such that $w$ directly precedes $v$ in $(X^v \cap V) \cup
  \widetilde{W} \cup \{v\}$ under $L'$:

  \begin{equation}\label{k-app-ccav-sum}
    f(v,z,s,\scores) = \sum_{w \in \mathcal{P}_v} f(w,z-z',s-s',\scores - \approval(w)),
  \end{equation}
  where (1) $z' = 1$ if $p \in \best_k(v)$ and $z'=0$ otherwise, and
  (2) $s' = 1$ if $v \in W$ and $s' = 0$ otherwise.  By combining
  Eq.~\eqref{k-app-ccav-transform} and \eqref{k-app-ccav-sum} we get:
  \begin{equation}\label{k-app-ccav-sum-final}
    f(v,z,s,\scores) = \sum_{w \in \mathcal{P}_v} f(w,z-z',s-s',\cutoff{\scores - \approval(w)}{\rank(w)}).
  \end{equation}

  We claim that function $g$ can be computed through
  Eq.~\eqref{k-app-ccav-convert} in polynomial time, using standard
  dynamic programming techniques to compute function $f$.  The reason
  is that to compute $f$ for the arguments as in
  Eq.~\eqref{k-app-ccav-convert} using recursive formula
  \eqref{k-app-ccav-sum-final}, we need to obtain $f$'s values for at
  most $\|C\| (\|V\| + \|W\|)^{k+3}$ different arguments. This is
  because starting from $\scores = \Gamma^{z-1}$, the only allowed
  transformations of $\scores$ are given in
  Eq.~\eqref{k-app-ccav-sum-final} and ensure that whenever we need to
  compute $f$, the $\scores$ vector is of the form
  $(z-1,z-1,...,z-1,e_1,e_2,...,e_k,0,0,...,0)$, where $e =
  (e_1,e_2,...,e_k)$ is some $k$-element vector of integers and for
  each $i$, $1 \le i \le k$, we have $0 \le e_i \le z - 1$.  Clearly,
  there are no more than $\|C\|z^k$ vectors of this form.  Thus, we
  can compute $g$ in polynomial time and, through
  Eq.~\eqref{k-app-ccav-main}, we can solve $k$-Approval-\#CCAV in
  polynomial time.
\end{proof}

\subsection{Approval Voting and Condorcet Voting}
Let us now consider Approval voting and Condorcet voting. While these
two systems may seem very different in spirit, their behavior with
respect to election control is similar. Specifically, in the
unrestricted case, for both systems \#CCAV and \#CCDV are
$\sharpp$-complete, for both systems it is impossible to make some
candidate a winner by adding candidates, and for both systems it is
impossible to prevent someone from winning by deleting
candidates. Yet, for both systems \#DCAC and \#CCDC are in $\fp$ via
almost identical algorithms. There is, however, also one difference.
For the single-peaked case, we were able to find polynomial-time
algorithms for voter control under Condorcet, while the results for
Approval remain elusive. 

\begin{theorem}\label{thm:approval-condorcet-avdv}
  In the unrestricted case, each of Approval-\#CCAV, Approval-\#CCDV,
  Condorcet-\#CCAV, and Condorcet-\#CCDV is $\sharpp$-complete. Their
  destructive variants are $\sharpp$-metric-complete.
\end{theorem}
\begin{proof}
  Note that the results for the destructive variants will follow by
  Theorem~\ref{thm:c-to-d} and so we focus on constructive variants
  only.

  For the case of Approval, $\sharpp$-completeness of \#CCAV and
  \#CCDV follows from the $\np$-completeness proofs for their decision
  variants given by Hemaspaandra, Hemaspaandra, and
  Rothe~\shortcite{hem-hem-rot:j:destructive-control}; these proofs,
  in effect, give parsimonious reductions from \#X3C to respective
  control problems and, thus, establish $\sharpp$-completeness.

  For the case of Condorcet, $\sharpp$-completeness of \#CCDV follows
  from the proofs of Theorems 5.1 and 4.19 of Faliszewski et
  al.~\cite{fal-hem-hem-rot:j:llull}, who effectively give a
  parsimonious reduction from \#X3C to Condorcet-\#CCDV. The case of
  Condorcet-\#CCAV appears to not have been considered in the
  literature and thus we give a direct proof (naturally, we could
  obtain $\sharpp$-Turing-completeness by noting that
  Theorem~\ref{thm:d-to-a} gives a Turing reduction from
  Condorcet-\#CCDV to Condorcet-\#CCAV, but $\sharpp$-completeness is
  a stronger result). For the reminder of the proof we focus on
  Condorcet-\#CCAV.  The problem is clearly in $\sharpp$, so it
  suffices to show that it is $\sharpp$-hard.  We give a parsimonious
  reduction from \#X3C.

  Let $(B,\mathcal{S})$ be an instance of \#X3C problem, where
  $B=\{b_1,\dots,b_{3k}\}$ and $\mathcal{S}=\{S_1,\dots,S_n\}$.  We
  create an election $E=(C,V)$, where $C=B\cup\{p\}$.  Let $V$ consist
  of $k-3$ voters with preferences $b_1\succ b_2\succ\cdots\succ
  b_{3k}\succ p$.  Thus, $b_1$ is the Condorcet winner of $E$, and
  every candidate $b_i\in B$ beats $p$ in $k-3$ votes.

  For each set $S_j\in\mathcal{S}$, $S_j =
  \{b_{j_1},b_{j_2},b_{j_3}\}$, let $W$ contain a voter with
  preference order $b_{j_1}\succ b_{j_2}\succ b_{j_3}\succ
  p\succ\cdots$ (after $p$ the remaining candidates are ranked in
  arbitrary order).  We claim that every subset $W'$, $W'\subseteq W$,
  such that $\|W'\|\le k$ and that $p$ is a Condorcet winner of election
  $E'=(C,V\cup W')$ corresponds one-to-one to a $k$-element subfamily
  $\calS'$ of $\calS$ whose elements union up to $B$ (i.e., $\calS'$
  is an exact set cover of $B$).

  First assume that there is a subfamily
  $\mathcal{S'}\subseteq\mathcal{S}$ which is an exact set cover of
  $B$.  For each $S_j\in\mathcal{S'}$, we include the corresponding
  voter from $W$ in $W'$. Let us consider election $E'=(C,V \cup W')$.
  For each $b_i\in B$ we have $\N{E'}(b_i,p)=\N{E}(b_i,p)+1=k-2$ and
  $\N{E'}(p,b_i)=\N{E}(p,b_i)+k-1=k-1$.  Thus $p$ becomes the
  Condorcet winner of $E'$.

  Now assume that $p$ is the Condorcet winner in election $E'=(C,V
  \cup W')$, where $W' \subseteq W$ and $\|W'\| \leq k$.  For each $b_i
  \in B$, there can be at most one voter in $W'$ that prefers $b_i$ to
  $p$. This is so, because otherwise we would have
  $\N{E'}(b_i,p)\ge\N{E}(b_i,p)+2=k-1$, and
  $\N{E'}(p,b_i)\le\N{E}(p,b_i)+k-2=k-2$, and so $p$ would lose to
  $b_i$.  Thus, each $b_i$ is preferred to $p$ by either zero or one
  voter from $W'$.  If $b_i$ is preferred by one voter from $W'$, then
  for $p$ to win, $p$ must be preferred to $b_i$ by $k-1$ voters from
  $W'$, and since some voter must be added, it must hold that
  $\|W'\|=k$.

  If there are no voters in $W'$ who prefer $b_i$ to $p$, then since
  each vote in $W'$ has some three candidates from $B$ ranked ahead of
  $p$, there must be some other $b_{i'}$ that is ranked above $p$ by
  more than one voter.  This contradicts the requirement that for each
  $b_j \in B$, at most one voter in $W'$ prefers $b_j$ to $p$.  Hence,
  each $b_i$ is preferred to $c$ by exactly one of the $k$ voters in
  $W'$.  Thus, the voters from $W'$ correspond to an exact set cover
  of $B$.
\end{proof}

For the case of single-peaked preferences, we can use dynamic
programming to solve voter control problems under Condorcet.

\begin{theorem}
  For the single-peaked case, Condorcet-\#CCAV, Condorcet-\#CCDV,
  Condorcet-\#DCAV, and Condorcet-\#DCDV are in \fp.
\end{theorem}
\begin{proof}
  We focus on Condorcet-\#CCAV. The remaining cases follow by applying
  Theorems~\ref{thm:c-to-d} and~\ref{thm:d-to-a}.

  We are given an election $E = (C,V)$ where $C$ is a set of
  candidates and $V$ is a set of voters, a designated candidate $p \in
  C$, a collection $W$ of unregistered voters, a nonnegative integer
  $k$, and an order $\mathrel{L}$, the societal axis, such that $V$
  and $W$ are single-peaked with respect to $\mathrel{L}$.  We show a
  polynomial-time algorithm for determining the number of sets $W'$,
  $W' \subseteq W$, where $\|W'\| \le k$, such that $p$ is a Condorcet
  winner in election $(C, V \cup W')$.  Our algorithm is based on the
  famous median-voter theorem.

  We split $C$ into three sets, $C_a = \{c \mid c \in C \land c
  \mathrel{L} p\}$, $C_b = \{c \mid c \in C \land p \mathrel{L} c\}$
  and $C_m = \{p\}$; $C_a$ contains the candidates that are before $p$ on
  the societal axis and $C_b$ contains the candidates that are after
  $p$. For each voter $v$, by $c_v$ we mean the candidate that $v$
  ranks first.  For each collection $U$ of voters from $V \cup W$, we
  define $U_a = \{v \mid v \in U \land c_v \in C_a\}$, $U_b = \{v \mid
  v \in U \land c_v \in C_b\}$ and $U_m = \{v \mid v \in U \land c_v =
  p\}$.  In other words, $U_a$ and $U_b$ consist of those voters from $U$
  for which the most preferred candidate is, respectively, in $C_a$ or
  $C_b$, and $U_m$ contains those voters from $U$ that rank $p$ first.
  
  For each $W' \subseteq W$, we define $\delta(W')$ to be $1$ exactly
  if the following conditions hold:
  \begin{enumerate}
  \item $\|V_a\|+\|W'_a\| < \frac{1}{2}(\|V\|+\|W'\|)$, and
  \item $\|V_b\|+\|W'_b\| < \frac{1}{2}(\|V\|+\|W'\|)$.
  \end{enumerate}
  Otherwise, we define $\delta(W') = 0$.  The following lemma is an
  expression of the well-known median voter theorem (we provide the
  short proof for the sake of completeness).

  \begin{lemma}
    For each $W' \subseteq W$, $p$ is the Condorcet winner of election
    $(C,V \cup W')$ if and only if $\delta(W') = 1$.
  \end{lemma}
  \begin{proof}
    Assume that $\delta(W') = 1$ and consider an arbitrary candidate
    $c$ other than $p$. For the sake of concreteness let us assume
    that $c \in C_a$, but a symmetric argument holds for $c \in C_b$.
    Due to single-peakedness of $V \cup W'$, no voter outside of $V_a$
    and $W'_a$ prefers $c$ to $p$. However, since $\delta(W')=1$, we
    have that $\|V_a\|+\|W'_a\| < \frac{1}{2}(\|V\|+\|W'\|)$, and so a
    majority of the voters prefers $p$ to $c$. Since $c$ was chosen
    arbitrarily, it holds that $p$ is a Condorcet winner.

    For the other direction, assume that $\delta(W') = 0$. For the
    sake of concreteness, let us assume that this is because
    $\|V_a\|+\|W'_a\| \geq \frac{1}{2}(\|V\|+\|W\|)$. By this
    assumption, there must be a candidate $c \in C$ that directly
    precedes $p$ with respect to $L$ (that is, $c \spl p$ and there is
    no candidate $d$ such that $c \spl d \spl p$). Due to
    single-peakedness of voters' preference orders, we have that every
    voter in $V_a \cup W'_a$ prefers $c$ to $p$, and so $p$ is not a
    Condorcet winner. A symmetric argument holds if $\|V_b\|+\|W'_b\|
    \geq \frac{1}{2}(\|V\|+\|W\|)$.
  \end{proof}

  Thus our algorithm should output the number of sets $W'$, $W'
  \subseteq W$, of cardinality at most $k$, such that $\delta(W')=1$
  holds. However, to evaluate $\delta(W')$ it suffices to know the
  values $\|W'_a\|$, $\|W'_b\|$, and $\|W'_m\|$. For each three
  integers $\ell_a$, $\ell_b$, and $\ell_m$ we define
  $\gamma(\ell_a,\ell_b,\ell_m)$ to be $1$ exactly if the following
  two conditions hold (these conditions are analogous to those for
  $\delta$):
  \begin{enumerate}
  \item $\|V_a\|+\ell_a < \frac{1}{2}(\|V\|+\ell_a+\ell_b+\ell_m)$, and
  \item $\|V_b\|+\ell_b < \frac{1}{2}(\|V\|+\ell_a+\ell_b+\ell_m)$.
  \end{enumerate}
  Otherwise, $\gamma(\ell_a,\ell_b,\ell_m)=0$. It is easy to see that
  for each $W' \subseteq W$ we have $\delta(W') =
  \gamma(\|W'_a\|,\|W'_b\|,\|W'_m\|)$. Now it follows that the number
  of subsets $W' \subseteq W$ of cardinality at most $k$ such that
  $p$ is a Condorcet winner of election $(C,V\cup W')$ is exactly:
  \[
    \sum_{\ell=0,...,k} 
    \sum_{\substack{
        \ell_a, \ell_b, \ell_m \in \naturals \\ 
        \ell_a + \ell_b + \ell_m = \ell
      }}
    \gamma(\ell_a,\ell_b,\ell_m)
    {\|W_a\| \choose \ell_a}
    {\|W_b\| \choose \ell_b}
    {\|W_m\| \choose \ell_m}.
  \]
  We can evaluate this sum in polynomial-time. This completes the
  proof.
\end{proof}

For the case of candidate control, we get polynomial-time algorithms
even for the unrestricted case. 

\begin{theorem} \label{th:apdcc} Approval-\#CCDC, and
  Condorcet-\#CCDC are in $\fp$, even in the unrestricted case.
\end{theorem}
\begin{proof}
  Let us handle the case of Approval first.  Let $I = (C,V,p,k)$ be an
  instance of approval-\#CCDC. The only way to ensure that $p \in C$
  is the unique winner is to remove all candidates $c \in C-\{p\}$
  such that $\score{(C,V)}^a(c)\geq\score{(C,V)}^a(p)$. Such
  candidates can be found immediately. Let's assume that there are
  $k_0$ such candidates. After removing all of them, we can also
  remove $k-k_0$ or less of any remaining candidates other than $p$.
  Based on this observation we provide the following simple algorithm.
  \begin{codebox}
    \Procname{$\proc{approval-\#CCDC}(C,V,p,k)$} \li Let $k_0$ be the
    number of candidates $c \in C-\{p\}$,
    \zi	\>s.t.\ $\score{(C,V)}^a(c)\geq\score{(C,V)}(p)$. \\[-4mm]
    \li \Return $\sum_{i=0}^{k-k_0}\binom{\|C\|-k_0-1}{i}$
  \end{codebox}
  Clearly, the algorithm is correct and runs in polynomial-time.

  For the case of Condorcet voting, it suffices to note that if $p$ is
  to be a winner, we have to delete all candidates $c \in C - \{p\}$
  such that $\N{(C,V)}(p,c) \leq \N{(C,V)}(c,p)$. Thus, provided that
  we let $k_0$ be the number of candidates $c \in C - \{p\}$ such that
  $\N{(C,V)}(p,c) \leq \N{(C,V)}(c,p)$, the same algorithm as for the
  case of approval voting works for Condorcet voting.
\end{proof}

\begin{theorem} \label{th:apacd}
  Both Approval-\#DCAC and Condorcet-\#DCAC are in $\fp$.
\end{theorem}

\begin{proof}
  We first consider the case of approval voting.  Let $I =
  (C,A,V,p,k)$ be an instance of Approval-\#DCAC, where $C = \{p,c_1,
  \ldots, c_{m-1}\}$ is the set of registered candidates, $A = \{a_1,
  \ldots, a_{m'}\}$ is the set of additional candidates, $V$ is the
  set of voters, $p$ is the designated candidate, and $k$ is the upper
  bound on the number of candidates that we can add. We will give a
  polynomial-time algorithm that counts the number of
  up-to-$k$-element subsets $A'$ of $A$ such that $p$ is not the
  unique winner of election $(C \cup A',V)$.

  Let $A_0$ be the set of candidates in $A$ that are approved by at
  least as many voters as $p$ is. To ensure that $p$ is not the unique
  winner of the election (assuming $p$ is the unique winner prior to
  adding any candidates), it suffices to include at least one
  candidate from $A_0$.  Thus, we have the following algorithm.

\begin{codebox}
\Procname{$\proc{Approval-\#DCAC}(C,A,V,p,k)$}
\li	\If $p$ is not the unique winner of $(C,V)$ \\[-4mm] \label{apacd:checking_c}
\li		\Then \Return $\sum_{i=0}^k\binom{\|A\|}{i}$ \\[-4mm] \label{apacd:returning_result1}
		\End
\li	Let $A_0$ be the set of candidates $a_i \in A$\!,
\zi	\>s.t.\ $\score{(C\cup A,V)}^a(a_i)\ge\score{(C\cup A,V)}^a(p)$. \label{apacd:defining_k0}
\li	$\id{result}:=0$ \label{apacd:result_init}
\li	\For $j:=1$ \To $k$ \label{apacd:counting_loop}
\li		\Do $\id{result}:=\id{result}+\sum_{i=1}^{\min(\|A_0\|,j)}\binom{\|A_0\|}{i}\binom{\|A-A_0\|}{j-i}$ \label{apacd:counting}
		\End
\li	\Return $\id{result}$ \label{apacd:returning_result2}
\end{codebox}

The loop from line~\ref{apacd:counting_loop}, for every $j$, counts
the number of ways in which we can choose exactly $j$ candidates from
$A$; it can be done by first picking $i$ of the candidates in $A_0$
(who beat $p$), and then $j-i$ of the candidates in $A - A_0$.  It is
clear that the algorithm is correct and runs in polynomial time.

Let us now move on to the case of Condorcet voting. It is easy to see
that the same algorithm works correctly, provided that we make two
changes: (a) in the first two lines, instead of testing if $p$ is an
approval winner we need to test if $p$ is a Condorcet winner, and (b)
we redefine the set $A_0$ to be the set of candidates $a_i \in A$ such
that $\N{(C \cup A,V)}(p,a_i) \leq \N{(C \cup A,V)}(a_i,p)$. To see that these two
changes suffice, it is enough to note that to ensure that $p$ is not a
Condorcet winner of the election we have to have that either $p$
already is not a Condorcet winner (and then we can freely add any
number of candidates), or we have to add at least one candidate from
$A_0$.
\end{proof}

We conclude our discussion of Condorcet voting by noting that for the
single-peaked case all our results for Condorcet directly translate to
all Condorcet-consistent rules. The reason for this is that if voters'
preferences are single peaked, then there always exist weak Condorcet
winners. Whenever weak Condorcet winners exist, they are the sole
winners under Condorcet-consistent rules by definition. Since we focus
on the unique-winner model, if $R$ is a Condorcet-consistent rule, $E$
is a single-peaked, and $c$ is a candidate, then $c$ is the unique
$R$-winner of $E$ if and only if $c$ is the unique Condorcet winner of
$E$. In effect, we have the following corollary.

\begin{corollary}\label{cor:condorcet-consistent}
  Let $R$ be a Condrocet-consistent rule. For the single-peaked case,
  $R$-\#CCDC, $R$-\#DCAC, $R$-\#CCAV, $R$-\#CCDV, $R$-\#DCAV, and
  $R$-\#DCDV are in $\fp$.
\end{corollary}

\subsection{Maximin Voting}
For the case of Maximin, we consider the unrestricted case
only. Maximin is Condorcet-consistent and, thus, for the single-peaked
case we can use Corollary~\ref{cor:condorcet-consistent}.

The complexity of decision variants of control for Maximin was studied
by Faliszewski, Hemaspaandra, and
Hemaspaandra~\shortcite{fal-hem-hem:j:multimode}.  In particular, they
showed that under Maximin all voter control problems are
$\np$-complete and an easy adaptation of their proofs gives the
following theorem.

\begin{theorem}
  Maximin-\#CCAV and Maximin-\#CCDV are $\sharpp$-complete, and
  Maximin-\#DCAV and Maximin-\#DCDV are $\sharpp$-metric-complete.
\end{theorem}

On the other hand, among the candidate control problems for Maximin,
only Maximin-CCAC is $\np$-complete (DCAC, CCDC, and DCDC are in
$\p$). Still, this hardness of control by adding candidates translates
into the hardness of all the counting variants of candidate control.
\begin{theorem}
  Maximin-\#CCAC is $\sharpp$-complete and Maximin-\#DCAC is
  $\sharpp$-metric-complete.
\end{theorem}
\begin{proof}
  Faliszewski et al.~\cite{fal-hem-hem:j:multimode}'s proof that
  Maximin-CCAC is $\np$-complete can be used without change; their
  reduction from X3C to Maximin-CCAC is also correct as a parsimonious
  reduction from \#X3C to Maximin-\#CCAC. The result for
  Maximin-\#DCAC follow through Theorem~\ref{thm:c-to-d}.
\end{proof}

The cases of Maximin-\#CCDC and Maximin-\#DCDC are more complicated
and require new ideas because decision variants of these problems are
in $\p$.

\begin{theorem}
  Both Maximin-\#CCDC and Maximin-\#DCDC are $\sharpp$-Turing-complete.
\end{theorem}
\begin{proof}
  We consider the \#CCDC case first. Clearly, the
  problem belongs to $\sharpp$ and it remains to show hardness. We
  will do so by giving a Turing reduction from \#PerfectMatching.

  Let $G = (G(X),G(Y),G(E))$ be our input graph, where $G(X) = \{x_1,
  \ldots, x_n\}$ and $G(Y) = \{y_1, \ldots, y_n\}$ are sets of
  vertices, and $E = \{e_1, \ldots, e_m\}$ is the set of edges.  For
  each nonnegative integer $k$, define $g(k)$ to be the number of
  matchings in $G$ that contain exactly $k$ edges (e.g., $g(n)$ is the
  number of perfect matchings in $G$).

  We define the following election $E = (C,V)$.  We set $C = G(E) \cup
  S \cup B \cup \{p\}$, where $S = \{s_0, \ldots, s_{n}\}$ and $B =
  \{b_{i,j}^\ell \mid \text{$0 \leq \ell \leq n$, $i < j$, and $e_i$
    and $e_j$ share a vertex}\}$. To build voter collection $V$, for
  each two candidates $a, b \in C$, we define $v(a,b)$ to be a pair of
  voters with preference orders $a > b > C - \{a,b\}$ and
  $\revnot{C-\{a,b\}} > a > b$. We construct $V$ as follows:
  \begin{enumerate}
\fixenumerate
  \item For each $s_i \in S$, we add pair $v(s_i,p)$.
  \item For each $s_i \in S$, we add two pairs $v(s_i,s_{i+1})$,
    where $i+1$ is taken modulo $n+1$.
  \item For each $s_i \in S$ and each $e_t \in E$, we add two pairs
    $v(s_i,e_t)$.
  \item For each $e_i, e_j \in E$, $i < j$, where $e_i$ and $e_j$
    share a vertex, and for each $\ell$, $0 \leq \ell \leq n$, we add
    two pairs $v(e_i,b_{i,j}^\ell)$ and two pairs
    $v(e_j,b_{i,j}^\ell)$.
  \end{enumerate}
  Let $T$ be the total number of pairs $v(a,b)$, $a,b \in C$, included
  in $V$. By our construction, the following properties hold:
  \begin{enumerate}
\fixenumerate
  \item $\score{E}^m(p) = T-1$ and it is impossible to change the
    score of $p$ be deleting $n$ candidates or fewer (this is because
    there are $n+1$ candidates $s_i \in S$ such that $N_E(p,s_i) =
    T-1$.
  \item For each $s_i \in S$, $\score{E}^m(s_i) = T-2$, but deleting
    $s_{i-1}$ (where we take $i-1$ modulo $n+1$) increases the score
    of $s_i$ to $T+1$.
  \item For each $e_t \in E$, $\score{E}^m(e_t)$ is $T-2$.
  \item For each $b_{i,j}^\ell \in B$, $\score{E}^m(b_{i,j}^\ell) = T-2$ and
    it remains $T-2$ if we delete either $e_i$ or $e_j$, but it
    becomes $T$ if we delete both $e_i$ and $e_j$.
  \end{enumerate}
  Note that $p$ is the unique winner of $E$.  For each $k$, $0 \leq k
  \leq n$, we form instance $I(k) = (C,V,p,k)$ of Maximin-\#CCDC. We
  define $f(k) = \#I(k) - \#I(k-1)$. That is, $f(k)$ is the number of
  solutions for $I(k)$ where we delete exactly $k$ candidates. We
  claim that for each $k$, $1 \leq k \leq n$, it holds that $ f(k) =
  \sum_{j=0}^k {\|B\| \choose j} g(k-j).  $ Why is this so? First,
  note that by the listed-above properties of $E$, deleting any subset
  $C'$ of candidates from $C$ that contains some member of $S$
  prevents $p$ from being a winner. Thus, we can only delete subsets
  $C'$ of $C$ that contains candidates in $G(E) \cup B$. Let us fix a
  nonnegative integer $r$, $0 \leq r \leq n$. Let $C' \subseteq G(E)
  \cup B$ be such that $p$ is the unique Maximin winner of $E' =
  (C-C',V)$ and $\|C'\| = r$.  Let $r_B = \|C' \cap B\|$ and $r_{G(E)}
  = \|C' \cap G(E)\|$.  It must be the case that for each $e_i, e_j
  \in G(E)$, $i < j$, where $e_i$ and $e_j$ share a vertex, $C'$
  contains at most one of them. Otherwise, $E'$ would contain at least
  one of the candidates $b_{i,j}^\ell$, $0 \leq \ell \leq n$, and this
  candidate would have score higher than $p$. Thus, the candidates in
  $C' \cap G(E)$ correspond to a matching in $G$ of cardinality
  $r_{G(E)}$.  On the other hand, since $r_B \leq n$, $C \cap B$
  contains an arbitrary subset of $B$. Thus, there are exactly ${\|B\|
    \choose r_B} g(r_{G(E)})$ such sets $C'$. Our formula for
  $f(k)$ is correct.

  Now, using standard algebra (a process similar to Gauss
  elimination), it is easy to verify that given values $f(1), f(2),
  \ldots, f(n)$, it is possible to compute (in this order) $g(0),
  g(1), \ldots, g(n)$. Together with the fact that constructing each
  $I(k)$, $0 \leq k \leq n$, requires polynomial time with respect to
  the size of $G$, this proves that given oracle access to
  Maximin-\#CCDC, we can solve \#PerfectMatching. Thus, Maximin-\#CCDC
  is $\sharpp$-Turing-complete and, by Theorem~\ref{thm:c-to-d}, so is
  Maximin-\#DCDC.
\end{proof}

\section{Related Work}
\label{sec:related}

The focus of this paper is on the complexity of predicting election
winners for the case, where we are uncertain about the structure of
the election (the exact identities of candidates/voters that
participate), yet we have perfect knowledge of voters' preference
orders. However, our model is just one of many approaches to winner
prediction, which in various forms and shapes has been studied in the
literature for some years already.  For example, to model imperfect
knowledge regarding voters' preferences, Konczak and
Lang~\shortcite{kon-lan:c:incomplete-prefs} introduced the possible
winner problem, further studied by many other researchers (see, e.g.,
\cite{con-xia:j:possible-necessary-winners,bet-dor:j:possible-winner-dichotomy,bac-bet-fal:c:counting-pos-win,che-lan-mau-mon:c:possible-winners-adding,lan-mon-xia:c:new-alternatives-new-results}). In
the possible winner problem, each voter is represented via a partial
preference order and we ask if there is an extension of these partial
orders to total orders that ensures a given candidate's victory.
Bachrach, Betzler, and
Faliszewski~\shortcite{bac-bet-fal:c:counting-pos-win} extended the
model by considering counting variants of possible winner problems.
Namely, they asked for how many extensions of the votes a given
candidate wins, in effect obtaining the probability of the candidate's
victory. This is very similar to our approach, but there are also
important differences. In the work of Bachrach et al., we have full
knowledge regarding the identities of candidates and voters
participating in the election, but we are uncertain about voters'
preference orders. In our setting, we have full knowledge about
voters' preference orders, but we are uncertain about the identities
of candidates/voters participating in the election.

Another model of predicting election outcomes is that of Hazon et
al.~\cite{haz-aum-kra-woo:j:uncertain-election-outcomes}.  They
consider a situation where each voter is undecided regarding several
possible votes. That is, for each voter we are given several possible
preference orders and a probability distribution over these votes.
The question is, what is the probability that a designated candidate
wins.

From a technical standpoint, our research continues the line of work
on the complexity of control. This line of work was initiated by
Bartholdi, Tovey, and Trick~\cite{bar-tov-tri:j:control}, and then
continued by Hemaspaandra, Hemaspaandra, and
Rothe~\shortcite{hem-hem-rot:j:destructive-control} (who introduced
the destructive cases), by Meir et
al.~\shortcite{mei-pro-ros-zoh:j:multiwinner} (who considered
multiwinner rules and who generalized the idea of the constructive and
destructive cases), by Faliszewski, Hemaspaandra, and
Hemaspaandra~\shortcite{fal-hem-hem:j:multimode} (who introduced
multimode model of control), by Faliszewski, Hemaspaandra, and
Hemaspaandra~\cite{fal-hem-hem:c:weighted-control} (who were first to
consider control for weighted elections), by Rothe and
Schend~\cite{rot-sch:c:experimental-control} (who initiated the
empirical study of the complexity of control problems), and by many
other researchers, who provided results for specific voting rules and
who introduced various other novel means of studying control problems
(see, e.g., the following
papers~\cite{bet-uhl:j:parameterized-complecity-candidate-control,erd-now-rot:j:sp-av,erd-rot:c:fallback-voting,fen-liu-lua-zhu:j:parameterized-control,liu-zhu:j:maximin,men-sin:c:schultze-control,par-xia:strategic-schultze-ranked-pairs};
we also point the readers to the
survey~\cite{fal-hem-hem:j:cacm-survey}).

Single-peaked elections were studied for a long time in social choice
literature, but they gained popularity in the computational social
choice world fairly recently, mostly due to the papers of
Walsh~\cite{wal:c:uncertainty-in-preference-elicitation-aggregation}
and Conitzer~\cite{con:j:eliciting-singlepeaked}. Then, Faliszewski et
al.~\cite{fal-hem-hem-rot:j:single-peaked-preferences} and, later,
Brandt et al.~\cite{bra-bri-hem-hem:c:sp2} studied the complexity of
control problems for single-peaked elections. Recently, Faliszewski,
Hemaspaandra, and Hemaspaandra~\cite{fal-hem-hem:j:nearly-sp}
complemented this line of work by studying nearly single-peaked
profiles.

Going in a different direction, our work is very closely related to
the paper of Walsh and Xia~\shortcite{wal-xia:c:lot-based} on
lot-based elections. Walsh and Xia study a model of Venetian
elections, where a group of voters of a given size is randomly
selected from a group of eligible voters, and the votes are collected
from these selected voters only. In this setting, the problem of
computing a candidate's chances of victory, in essence, boils down to
the counting variant of control by adding voters problem. Thus, our
paper and that of Walsh and Xia are quite similar on the technical
front. The papers, however, have no overlap in terms of results.

\section{Conclusions and Future Work}\label{sec:conclusions}
We have considered a model of predicting election winners in
settings where there is uncertainty regarding the structure of the
election (that is, regarding the exact set of candidates and the exact
collection of voters participating in the election).  We have shown
that our model corresponds to the counting variants of election
control problems (specifically, we have focused on election control by
adding/deleting candidates and voters). We have considered Plurality,
Approval, Condorcet, $k$-Approval, and Maximin (see
Table~\ref{tab:results} for our results).  For the former three, the
complexity of counting variants of control is analogous to the
complexity of decision variants of respective problems, but for the
latter two, some of the counting control problems are more
computationally demanding than their decision counterparts.

Many of our results indicate computational hardness of winner
prediction problems.  To alleviate this issue to some extent, we also
considered single-peaked preferences that are more likely to appear in
practice. In this case, we got polynomial-time results only (except
the case of Approval, where have no results for the single-peaked
case). Still, sometimes in practice one might have to seek heuristic
algorithms or approximate solutions (e.g., sampling-based algorithms
similar to the one of Bachrach, Betzler, and
Faliszewski~\cite[Theorem~6]{bac-bet-fal:c:counting-pos-win}).

There are many ways to extend our work. For example, in the
intruduction we mentioned the model where for each voter $v$ (or
candidate $c$) we have probability $p_v$ (probability $p_c$) that this
voter (candidate) participates in the election. We believe that
studying this problem in more detail would be very interesting.  As we
have argued, the model where all values $p_v$ ($p_c$) are identical,
reduces to our setting, but the cases where the probabilities can
differ remain open.

\medskip
\noindent\textbf{Acknowledgements.} 
An early version of this paper appeared in IJCAI Workshop on Social
Choice and Artificial Intelligence (2011) and an extended abstract was
presented at the Twenty-Sixth AAAI Conference on Artificial
Intelligence (AAAI-12).  We thank both the AAAI and WSCAI reviewers
for very helpful feedback.  Piotr Faliszewski was in part supported by
AGH University of Technology Grant 11.11.230.015 (statutory project),
by Foundation for Polish Science's program Homing/Powroty, and by
Poland's National Science Center's grant 2012/06/M/ST1/00358.

\label{sec:summary}

\bibliographystyle{plain}
\bibliography{grypiotr2006}

\begin{thebibliography}{10}

\bibitem{bac-bet-fal:c:counting-pos-win}
Y.~Bachrach, N.~Betzler, and P.~Faliszewski.
\newblock Probabilistic possible winner determination.
\newblock In {\em Proceedings of AAAI-10}, pages 697--702. AAAI Press, July
  2010.

\bibitem{bal-har:j:characterization-single-peaked}
M.~Ballester and G.~Haeringer.
\newblock A characterization of the single-peaked domain.
\newblock {\em Social Choice and Welfare}, 36(2):305--322, 2011.

\bibitem{bar-tov-tri:j:control}
J.~{{Bartholdi}}, III, C.~Tovey, and M.~Trick.
\newblock How hard is it to control an election?
\newblock {\em Mathematical and Computer Modeling}, 16(8/9):27--40, 1992.

\bibitem{bar-tri:j:stable-matching-from-psychological-model}
J.~{{Bartholdi}}, III and M.~Trick.
\newblock Stable matching with preferences derived from a psychological model.
\newblock {\em Operations Research Letters}, 5(4):165--169, 1986.

\bibitem{bet-dor:j:possible-winner-dichotomy}
N.~Betzler and B.~Dorn.
\newblock Towards a dichotomy of finding possible winners in elections based on
  scoring rules.
\newblock {\em Journal of Computer and System Sciences}, 76(8):812--836, 2010.

\bibitem{bet-uhl:j:parameterized-complecity-candidate-control}
N.~Betzler and J.~Uhlmann.
\newblock Parameterized complexity of candidate control in elections and
  related digraph problems.
\newblock {\em Theoretical Computer Science}, 410(52):43--53, 2009.

\bibitem{bla:b:polsci:committees-elections}
D.~Black.
\newblock {\em The Theory of Committees and Elections}.
\newblock Cambridge University Press, 1958.

\bibitem{bra-bri-hem-hem:c:sp2}
F.~Brandt, M.~Brill, E.~Hemaspaandra, and L.~Hemaspaandra.
\newblock Bypassing combinatorial protections: Polynomial-time algorithms for
  single-peaked electorates.
\newblock In {\em Proceedings of AAAI-10}, pages 715--722. AAAI Press, July
  2010.

\bibitem{bra-con-end:b:comsoc}
F.~Brandt, V.~Conitzer, and U.~Endriss.
\newblock Computational social choice.
\newblock In G.~Wei\ss, editor, {\em Multiagent Systems}, pages 213--283. MIT
  Press, 2013.

\bibitem{che-lan-mau-mon:c:possible-winners-adding}
Y.~Chevaleyre, J.~Lang, N.~Maudet, and J.~Monnot.
\newblock Possible winners when new candidates are added: {T}he case of scoring
  rules.
\newblock In {\em Proceedings of AAAI-10}, pages 762--767. AAAI Press, July
  2010.

\bibitem{con:j:eliciting-singlepeaked}
V.~Conitzer.
\newblock Eliciting single-peaked preferences using comparison queries.
\newblock {\em Journal of Artificial Intelligence Research}, 35:161--191, 2009.

\bibitem{dwo-kum-nao-siv:c:rank-aggregation}
C.~Dwork, R.~Kumar, M.~Naor, and D.~Sivakumar.
\newblock Rank aggregation methods for the web.
\newblock In {\em Proceedings of the 10th International World Wide Web
  Conference}, pages 613--622. ACM Press, March 2001.

\bibitem{elk-fal-sli:j:cloning}
E.~Elkind, P.~Faliszewski, and A.~Slinko.
\newblock Cloning in elections: {F}inding the possible winners.
\newblock {\em Journal of Artificial Intelligence Research}, 42:529--573, 2011.

\bibitem{eph-ros:j:multiagent-planning}
E.~Ephrati and J.~Rosenschein.
\newblock A heuristic technique for multi-agent planning.
\newblock {\em Annals of Mathematics and Artificial Intelligence},
  20(1--4):13--67, 1997.

\bibitem{erd-now-rot:j:sp-av}
G.~Erd\'{e}lyi, M.~Nowak, and J.~Rothe.
\newblock Sincere-strategy preference-based approval voting fully resists
  constructive control and broadly resists destructive control.
\newblock {\em Mathematical Logic Quarterly}, 55(4):425--443, 2009.

\bibitem{erd-rot:c:fallback-voting}
G.~Erd\'{e}lyi and J.~Rothe.
\newblock Control complexity in fallback voting.
\newblock In {\em Proceedings of Computing: the 16th Australasian Theory
  Symposium}, pages 39--48. Australian Computer Society {\it Conferences in
  Research and Practice in Information Technology Series}, vol.~32, no.~8,
  January 2010.

\bibitem{esc-lan-ozt:c:single-peaked-consistency}
B.~Escoffier, J.~Lang, and M.~{\"O}zt{\"u}rk.
\newblock Single-peaked consistency and its complexity.
\newblock In {\em Proceedings of the 18th European Conference on Artificial
  Intelligence}, pages 366--370. IOS Press, July 2008.

\bibitem{fal-hem-hem:j:cacm-survey}
P.~Faliszewski, E.~Hemaspaandra, and L.~Hemaspaandra.
\newblock Using complexity to protect elections.
\newblock {\em Communications of the ACM}, 53(11):74--82, 2010.

\bibitem{fal-hem-hem:j:multimode}
P.~Faliszewski, E.~Hemaspaandra, and L.~Hemaspaandra.
\newblock Multimode control attacks on elections.
\newblock {\em Journal of Artificial Intelligence Research}, 40:305--351, 2011.

\bibitem{fal-hem-hem:c:weighted-control}
P.~Faliszewski, E.~Hemaspaandra, and L.~Hemaspaandra.
\newblock Weighted electoral control.
\newblock In {\em Proceedings of AAMAS-13}, pages 367--374. International
  Foundation for Autonomous Agents and Multiagent Systems, May 2013.

\bibitem{fal-hem-hem:j:nearly-sp}
P.~Faliszewski, E.~Hemaspaandra, and L.~Hemaspaandra.
\newblock The complexity of manipulative attacks in nearly single-peaked
  electorates.
\newblock {\em Artificial Intelligence}, 207:69--99, 2014.

\bibitem{fal-hem-hem-rot:j:llull}
P.~Faliszewski, E.~Hemaspaandra, L.~Hemaspaandra, and J.~Rothe.
\newblock Llull and {Copeland} voting computationally resist bribery and
  constructive control.
\newblock {\em Journal of Artificial Intelligence Research}, 35:275--341, 2009.

\bibitem{fal-hem-hem-rot:j:single-peaked-preferences}
P.~Faliszewski, E.~Hemaspaandra, L.~Hemaspaandra, and J.~Rothe.
\newblock The shield that never was: {Societies} with single-peaked preferences
  are more open to manipulation and control.
\newblock {\em Information and Computation}, 209(2):89--107, 2011.

\bibitem{fal-pro:j:manipulation}
P.~Faliszewski and A.~Procaccia.
\newblock {AI}'s war on manipulation: {A}re we winning?
\newblock {\em AI Magazine}, 31(4):52--64, 2010.

\bibitem{gho-mun-her-sen:c:voting-for-movies}
S.~Ghosh, M.~Mundhe, K.~Hernandez, and S.~Sen.
\newblock Voting for movies: {T}he anatomy of recommender systems.
\newblock In {\em Proceedings of the 3rd Annual Conference on Autonomous
  Agents}, pages 434--435. ACM Press, 1999.

\bibitem{haz-aum-kra-woo:j:uncertain-election-outcomes}
N.~Hazon, Y.~Aumann, S.~Kraus, and M.~Wooldridge.
\newblock On the evaluation of election outcomes under uncertainty.
\newblock {\em Artificial Intelligence}, 189:1--18, 2012.

\bibitem{hem-hem-men:c:search-decision}
E.~Hemaspaandra, L.~Hemaspaandra, and C.~Menton.
\newblock Search versus decision for election manipulation problems.
\newblock In {\em Proceedings of the 30th Annual Symposium on Theoretical
  Aspects of Computer Science}, pages 377--388. Leibniz-Zentrum f{\"u}r
  Informatik, 2013.

\bibitem{hem-hem-rot:j:destructive-control}
E.~Hemaspaandra, L.~Hemaspaandra, and J.~Rothe.
\newblock Anyone but him: {The} complexity of precluding an alternative.
\newblock {\em Artificial Intelligence}, 171(5--6):255--285, 2007.

\bibitem{hun-mar-rad-ste-:j:planar-counting-problems}
H.~Hunt, M.~Marathe, V.~Radhakrishnan, and R.~Stearns.
\newblock The complexity of planar counting problems.
\newblock {\em SIAM Journal on Computing}, 27(4):1142--1167, 1998.

\bibitem{kon-lan:c:incomplete-prefs}
K.~Konczak and J.~Lang.
\newblock Voting procedures with incomplete preferences.
\newblock In {\em Proceedins of the Multidisciplinary IJCAI-05 Worshop on
  Advances in Preference Handling}, pages 124--129, July/August 2005.

\bibitem{kre:j:optimization}
M.~Krentel.
\newblock The complexity of optimization problems.
\newblock {\em Journal of Computer and System Sciences}, 36(3):490--509, 1988.

\bibitem{kut-kit:j:nlp}
M.~Kuta and J.~Kitowski.
\newblock Benchmarking high performance architectures with natural language
  processing algorithms.
\newblock {\em Computer Science}, 12:19--31, 2011.

\bibitem{lin:thesis:elections}
A.~Lin.
\newblock {\em Solving Hard Problems in Election Systems}.
\newblock PhD thesis, Rochester Institute of Technology, Rochester, NY, 2012.

\bibitem{fen-liu-lua-zhu:j:parameterized-control}
H.~Liu, H.~Feng, D.~Zhu, and J.~Luan.
\newblock Parameterized computational complexity of control problems in voting
  systems.
\newblock {\em Theoretical Computer Science}, 410(27--29):2746--2753, 2009.

\bibitem{liu-zhu:j:maximin}
H.~Liu and D.~Zhu.
\newblock Parameterized complexity of control problems in maximin election.
\newblock {\em Information Processing Letters}, 110(10):383--388, 2010.

\bibitem{bou-lu:c:chamberlin-courant}
T.~Lu and C.~Boutilier.
\newblock Budgeted social choice: From consensus to personalized decision
  making.
\newblock In {\em Proceedings of IJCAI-11}, pages 280--286, 2011.

\bibitem{mei-pro-ros-zoh:j:multiwinner}
R.~Meir, A.~Procaccia, J.~Rosenschein, and A.~Zohar.
\newblock The complexity of strategic behavior in multi-winner elections.
\newblock {\em Journal of Artificial Intelligence Research}, 33:149--178, 2008.

\bibitem{men-sin:c:schultze-control}
C.~Menton and P.~Singh.
\newblock Control complexity of {Schulze} voting.
\newblock In {\em Proceedings of IJCAI-13}, pages 286--292, 2013.

\bibitem{mes-nur:b:distance-realizability}
T.~Meskanen and H.~Nurmi.
\newblock Closeness counts in social choice.
\newblock In M.~Braham and F.~Steffen, editors, {\em Power, Freedom, and
  Voting}. Springer-Verlag, 2008.

\bibitem{mia:t:priced-control}
T.~Mi\k{a}sko.
\newblock Algorithms and complexity results for election control problems with
  prices.
\newblock Master's thesis, AGH University of Science and Technology, Krak\'ow,
  Poland, September 2013.

\bibitem{obr-elk:c:random-ties-matter}
S.~Obraztsova and E.~Elkind.
\newblock On the complexity of voting manipulation under randomized
  tie-breaking.
\newblock In {\em Proceedings of IJCAI-11}, pages 319--324, July 2011.

\bibitem{obr-elk-haz:c:ties-matter}
S.~Obraztsova, E.~Elkind, and N.~Hazon.
\newblock Ties matter: {C}omplexity of voting manipulation revisited.
\newblock In {\em Proceedings of the 10th International Conference on
  Autonomous Agents and Multiagent Systems}, pages 71--78, May 2011.

\bibitem{obr-zic-elk:c:multiwinner-tie-breaking}
S.~Obraztsova, Y.~Zick, and E.~Elkind.
\newblock On manipulation in multi-winner elections based on scoring rules.
\newblock In {\em Proceedings of AAMAS-13}, pages 359--366, 2013.

\bibitem{pap:b:complexity}
C.~Papadimitriou.
\newblock {\em Computational Complexity}.
\newblock Addison-Wesley, 1994.

\bibitem{par-xia:strategic-schultze-ranked-pairs}
D.~Parkes and L.~Xia.
\newblock A complexity-of-strategic-behavior comparison between {S}chulze's
  rule and ranked pairs.
\newblock In {\em Proceedings of AAAI-12}, pages 1429--1435, July 2012.

\bibitem{rot-sch:c:experimental-control}
J.~Rothe and L.~Schend.
\newblock Control complexity in {B}ucklin, fallback, and plurality voting: {A}n
  experimental approach.
\newblock In {\em Proceedings of the 11th International Symposium on
  Experimental Algorithms}, pages 356--368. Springer-Verlag {\it Lecture Notes
  in Computer Science \#7276}, June 2012.

\bibitem{sim:thesis:complexity}
J.~Simon.
\newblock {\em On Some Central Problems in Computational Complexity}.
\newblock PhD thesis, Cornell University, Ithaca, N.Y., January 1975.
\newblock Available as Cornell Department of Computer Science Technical Report
  TR75-224.

\bibitem{sko-fal-sli:c:multiwinner}
P.~Skowron, P.~Faliszewski, and A.~Slinko.
\newblock Fully proportional representation as resource allocation:
  {A}pproximability results.
\newblock In {\em Proceedings of IJCAI-13}, pages 353--359, 2013.

\bibitem{sui-fra-bou:c:single-peaked}
X.~Sui, A.~Francois-Nienaber, and C.~Boutilier.
\newblock Multi-dimensional single-peaked consistency and its approximations.
\newblock In {\em Proceedings of IJCAI-13}, pages 375--382, 2010.

\bibitem{val:j:permanent}
L.~Valiant.
\newblock The complexity of computing the permanent.
\newblock {\em Theoretical Computer Science}, 8(2):189--201, 1979.

\bibitem{wal:c:uncertainty-in-preference-elicitation-aggregation}
T.~Walsh.
\newblock Uncertainty in preference elicitation and aggregation.
\newblock In {\em Proceedings of AAAI-07}, pages 3--8. AAAI Press, July 2007.

\bibitem{wal-xia:c:lot-based}
T.~Walsh and L.~Xia.
\newblock Lot-based voting rules.
\newblock In {\em Proceedings of AAMAS-12}, pages 603--610, June 2012.

\bibitem{con-xia:j:possible-necessary-winners}
L.~Xia and V.~Conitzer.
\newblock Determining possible and necessary winners given partial orders.
\newblock {\em Journal of Artificial Intelligence Research}, 41:25--67, 2011.

\bibitem{lan-mon-xia:c:new-alternatives-new-results}
L.~Xia, J.~Lang, and J.~Monnot.
\newblock Possible winners when new alternatives join: {N}ew results coming up!
\newblock In {\em Proceedings of the 10th International Conference on
  Autonomous Agents and Multiagent Systems}, pages 829--836. International
  Foundation for Autonomous Agents and Multiagent Systems, May 2011.

\bibitem{zan:j:sharp-p}
V.~Zank{\'{o}}.
\newblock \#{P}-completeness via many-one reductions.
\newblock {\em International Journal of Foundations of Computer Science},
  2(1):76--82, 1991.

\end{thebibliography}

\end{document}